\documentclass{article}
 \usepackage{epsfig}
 \usepackage{amsmath}
 \usepackage{graphics}
 \usepackage{latexsym}
 \usepackage{amssymb}

 \title{On the Convergence of the Convectively Filtered Burgers Equation to the Entropy Solution of the Inviscid Burgers Equation}

 \author{
 Greg Norgard %
    \thanks{Graduate Student,Department of Applied Mathematics.}
 and Kamran Mohseni%
   \thanks{Associate Professor of Aerospace Engineering Sciences; Affiliated faculty in the Applied Mathematics Department.}\\
  {\normalsize\itshape
  University of Colorado, Boulder, Colorado, 80309, US}
 }

 \def\p{\partial}
 \def\ubar{\bar{u}}

\def\spacce#1{\hskip #1pt}
\def\drawline#1#2{\raise 2.5pt\vbox{\hrule width #1pt height #2pt}}
\def\solid{\drawline{24}{.5}\nobreak}
\def\bdash{\hbox{\drawline{4}{.5}\spacce{2}}}
\def\dashed{\bdash\bdash\bdash\bdash\hskip-2pt\nobreak}
\def\bdot{\hbox{\drawline{1}{.5}\spacce{2}}}

\def\dashdot{\bdash\bdot\bdash\bdot\bdash\bdot}
\def\dashdotdot{\bdash\bdot\bdot\bdash\bdot\bdot\bdash\bdot\bdot}

\newtheorem{theorem}{Theorem}[section]
\newtheorem{lemma}[theorem]{Lemma}

\newtheorem{conjecture}[theorem]{Conjecture}
\newtheorem{definition}[theorem]{Definition}

\newenvironment{proof}[1][Proof]{\begin{trivlist}
\item[\hskip \labelsep {\bfseries #1}]}{\end{trivlist}}

\begin{document}

\maketitle

\begin{abstract}
This document provides a proof that the solutions to the
convectively filtered Burgers equation, will converge to the entropy
solution of the inviscid Burgers equation when certain restrictions
are put on the initial conditions.  It does so by first establishing
convergence to a weak solution of the inviscid Burgers equation and
then showing that the weak solution is the entropy solution.  Then
the results are extended to encompass more general initial
conditions.
\end{abstract}

\section{Introduction}
Using a filtered velocity in fluid dynamics is not a new concept.
Filtered velocities have been used in turbulence modeling in Large
Eddy Simulation (LES) \cite{Germano:91a,Lesieur:96a,Hughes:01b},
Lagrangian Averaged Navier-Stokes (LANS-$\alpha$)
\cite{FoiasC:02a,Marsden:01h,Mohseni:03a,Chen:99b}, and Leray
turbulence modeling \cite{Holm:04a, IlyinAA:06a, Hanjalic:06a}.
Specifically in the LANS-$\alpha$ and Leray approaches, a filtered
velocity is used in the nonlinear term of the Navier-Stokes
equations.  A form of the compressible Euler equations with a
filtered velocity has also been developed using the Lagrangian
averaging \cite{Mohseni:03c}.  In our earlier paper
\cite{Norgard:08b}, it was discussed that it should be possible to
model both turbulence and shock formation using such a filtered
velocity.  This was motivated by realizing that turbulence and
shocks are both consequences of the nonlinear term and its resulting
cascade of energy into smaller scales.  Thus it should be possible
to capture both effects with proper small scale modeling.  It has
been seen that some turbulent behavior has been successfully modeled
using a filtered velocity in the LANS-$\alpha$ and Leray approaches.
This paper in conjunction with our previous paper
\cite{Norgard:08b},  aims at showing that such a technique can
successfully model shock formation.

The investigation begins with the inviscid Burgers equation,
\begin{equation}
 u_t+uu_x=0.
\end{equation}
Burgers equation was chosen because it shares the same nonlinear
term as the Euler and Navier-Stokes equations.  Additionally it is a
conservation law, like the Euler equations.  It is known to form
shocks, and has been well studied.

It is well established that the inviscid Burgers equation forms
discontinuities in finite time, determined by initial conditions
\cite{WhithamGB:74a, LaxPD:73a}.  To deal with these discontinuities
weak solutions are introduced.  However, when weak solutions are
introduced, solutions are no longer necessarily unique
\cite{LaxPD:73a, OleinikOA:57a}.  In order to choose the physically
relevant solution, an entropy condition is applied, which one and
only one weak solution satisfies. This physically relevant solution
is referred to as the entropy solution.  Lax, Oleinik, and Kruzkov
have examined the entropy condition for conservation laws and
expressed it using different techniques \cite{LaxPD:73a,
OleinikOA:57a, KruzkovSN:70a}.  Each of their entropy conditions can
be used in different classes of conservation laws, but can all be
applied to the inviscid Burgers equation with equivalent results
\cite{LellisCD}.  This paper uses the Lax entropy condition, which
is explained in section \ref{previouslyestablishedfacts}.

Classically the inviscid Burgers equation is regularized by adding
viscosity, resulting in the equation
\begin{equation}
 u_t+uu_x=\nu u_{xx}.
\end{equation}
This regularization has been proven to converge to the entropy
solution of the inviscid Burgers equation as $\nu \to 0$
\cite{LaxPD:73a, OleinikOA:57a, KruzkovSN:70a}.

This  paper considers the equations
\begin{subequations}\label{IVP2}
\begin{eqnarray}
\label{IVP2a}
u_t+\ubar u_x=0\\
\label{IVP2b}
\ubar=g^\alpha \ast u\\
\label{IVP2c} u(x,0)=u_0(x)
\end{eqnarray}
\end{subequations}
where
\begin{equation}
 g^\alpha=\frac{1}{\alpha} g\left( \frac{x}{\alpha} \right)
\end{equation}
where $g$ is a chosen filter.  These equations replace the
convective velocity of the inviscid Burgers equation with a filtered
velocity.  Thus, equations (\ref{IVP2a}) and (\ref{IVP2b}) are
referred to as the convectively filtered Burgers equation (CFB).
While it has been proven that the solutions to the CFB equations
exist \cite{Norgard:08b}, previously it has only been proven that
the solutions for the Helmholtz filter converge to a weak solution
of the inviscid Burgers equation with attempts to show numerically
convergence to the entropy solution \cite{BhatHS:06a}.

 This paper proves that for a specific set of initial conditions that the
solutions to the CFB equations will converge to the entropy solution
of the inviscid Burgers equation.  Specifically we will look at bell
shaped, continuously differentiable initial conditions rigorously
defined in Definition \ref{conditionB}.  We then give rationale and
make a conjecture on how the CFB equations will converge to the
entropy solution for any continuous initial conditions, and how to
regain an entropy solution for discontinuous initial conditions.

 The following section reviews established facts about the inviscid Burgers
equation and some of the recent work regarding the CFB equations.
Section \ref{weaksolutionsection} proves that solutions to the CFB
equations converge to a weak solution of the inviscid Burgers
equation, and section \ref{entropysolutionsection} proves
convergence to the entropy solution.  Section
\ref{moreentropysolutionsection} then extends the results of section
\ref{entropysolutionsection} and conjectures that it can be extended
further.  Section \ref{numericssection} runs some numerical
simulations and examines the results. All is the followed with
concluding remarks.

%
%
%

\section{Background Information of Burgers Equation and the CFB equations}\label{previouslyestablishedfacts}
Burgers equation has been thoroughly researched by many people over
the years.  This section provides a review of some of the previously
established properties of the inviscid Burgers equation.  Many of
these will be used later on to establish new results about the CFB
equations.  This section will also list some of the previously
established properties of the CFB equations, which are also crucial
to the analysis found in the following sections.

\subsection{Method of Characteristics}
The inviscid Burgers equation lends itself well to examination with
method of characteristics.  From Whitham \cite{WhithamGB:74a}, the
inviscid Burgers equation can be broken into two ODE's
\begin{align}
u_t(\xi)=0\\
\frac{\p}{\p t} \xi=u(\xi).
\end{align}
From this it is determined that along the characteristics
\begin{equation}
\xi=x_0+u_0(\xi)t
\end{equation}
$u(x)$ is constant. Thus characteristics travel at the speed equal
to the value of $u$ along those characteristics.  This is, true,
until characteristics cross, forming shocks. This is equivalent to
seeing that the material derivative is zero \cite{Marsden:94a}.

\subsection{Weak Solutions and Entropy Conditions}
Lax \cite{LaxPD:73a}  addresses weak solutions and entropy solutions
of conservation laws.  From his work, a lot of information can be
gained about the solutions to the inviscid Burgers equation.

The first thing we learn is that any weak solution to the inviscid
Burgers equation must satisfy the integral form of the conservation
law, or
\begin{equation}
\left. \int_{g}^{h} u \, dx  \right|_{t_1}^{t_2}=\left.
\int_{t_1}^{t_2} \frac{-u^2}{2} \right|_{g}^{h} \, dt,
\end{equation}
which must hold for any $g$ and $h$ and every time interval
$(t_1,t_2)$.  A consequence of this are the Rankine-Hugoniot jump
conditions. These dictate the speed at which any discontinuity can
propagate.  If $s$ is the position of a shock then,
\begin{equation}
 \frac{d}{dt} s(t)=\frac{1}{2}\left[  u(s^-) + u(s^+) \right].
 \end{equation}

Lax also establishes the existence and uniqueness of a weak solution
to the inviscid Burgers equation which satisfies the so called
entropy condition
\begin{equation}
 u(s^-) > u(s^+)
\end{equation}
where $s$ is the location of a discontinuity.  Thus the only
discontinuities that are allowed to exist in this ``entropy
solution'' are decreasing jumps.

Lax also states that for solutions satisfying the entropy condition,
``every point can be connected by a backward drawn characteristic to
a point on the initial line.''  Thus any value of the entropy
solution, $u(x,t)$, can be traced back to the initial conditions.
For discontinuous initial conditions, points traced back to the
point of discontinuity can take on values between the left and right
limits of the discontinuity as is shown in section
\ref{behaviorexample}.  For continuous initial conditions the
entropy solution can be written as $u(x,t)=u_0(\phi(x,t))$, where
$\phi(x,t)$ is an increasing function of $x$ for any time, and
$\phi(x,0)=x$.

Here we will define what will be referred to in this paper as a
reparameterization of a function.

\begin{definition}
If $\phi(x,t)$ is an increasing function of $x$ for any time, and
$\phi(x,0)=x$, the function $f(\phi(x,t))$ will be called a
reparameterization of the function $f$.
\end{definition}

It is clear that at any time $t$, a reparameterization of the
function $f$ cannot obtain values that are not obtained by $f$.  It
can, however, lack values that are found in $f$, as it was not
dictated that $\phi$ be onto for all time.  Looking back to the
previous paragraph we can see that for continuous initial
conditions, the entropy solution to the inviscid Burgers equation
will be a reparameterization of the initial conditions.

\subsection{Properties of the CFB equations}
From previous work by our group \cite{ Mohseni:06l, Mohseni:07e,
Mohseni:07s,Norgard:08b} the following theorem is established. It is
presented here in its one dimensional form.

\begin{theorem}
\label{existencetheorem} Let $g(x)$ $\in$ $W^{1,1}(\mathbb{R})$ and
$u_0(x)$ $\in$ $C^1(\mathbb{R})$, then there exists a unique global
solution $u(x,t)$ $\in$ $C^1(\mathbb{R},\mathbb{R})$ to the initial
value problem (\ref{IVP}).

\begin{subequations}\label{IVP}
\begin{align}
\label{IVPa}
u_t+\ubar u_x=0\\
\label{IVPb}
\ubar=g\ast u\\
\label{IVPc} u(x,0)=u_0(x)
\end{align}
\end{subequations}

\end{theorem}

A sketch of the theorem is as follows.  Examine the equations using
method of characteristics.  Due to the nature of the equations the
infinity norm of $u$ will be bounded for all time.  By Young's
inequality $||\ubar_x||_\infty$ can thus be bounded for all time.
The characteristics of the equations will not cross if their
Jacobian remains nonzero.  The rate of change of the Jacobian can be
directly related to $\ubar_x$ by
\begin{equation}
\frac{\partial}{\partial t} J =\ubar_x \, J.
\end{equation}
  Since $||\ubar_x||_\infty$
remains bounded, the Jacobian will remain nonzero, the
characteristics will not cross, and a unique solution will exist for
any finite time.

In the course of proving the theorem, it was established that the
solution take the form $u(x,t)=u_0(\phi(x,t))$ where $\phi(x,t)$ is
a continuous, invertible, and increasing function of $x$ for any
time, and $\phi(x,0)=x$.  Thus the solution is a reparameterization
of its initial conditions.

\section{Weak Solution}\label{weaksolutionsection}
Regularizations of conservation laws do not necessarily have to
converge to weak solutions to those conservation law.  Take for
example the KdV equations,
\begin{equation}
\label{kdv} u_t+u\,u_x=-\epsilon u_{xxx}.
\end{equation}
This regularizes the inviscid Burgers equation in the sense that
solutions are now continuous, however, many oscillations form as
$\epsilon \to 0$, requiring a weak limit for convergence
\cite{Lax:83a,Kawahara:70b}. This limit is not a weak solution of
the inviscid Burgers equation \cite{GurevichAV:74a}, and thus
definitely not the entropy solution.

Thus the first step to proving convergence to the entropy solution
is to prove convergence to a weak solution.  The following
subsections prove this by showing that a subsequence of the
solutions to the CFB equations must converge to a function in
$L^1_\text{loc}$.  It is then shown that this function is, in fact,
a weak solution to the inviscid Burgers equation.

\subsection{Convergence of Solutions}
In this section we show that the solutions of the CFB equations
($u^\alpha$) converge to a function $u$.  This subsection mirrors
work done by Bhat and Fetecau \cite{BhatHS:06a}.  We begin by
claiming the following properties of the solutions $u^\alpha$.

\begin{lemma}  The solutions to the initial value problem (\ref{IVP2}) have the following properties.
\begin{eqnarray}
\label{bounded}
||u^\alpha(\cdot,t)||_{L^\infty}=||u^\alpha(\cdot,0)||_{L^\infty} = ||u_0||_{L^\infty}=A_1,\\
\label{totalvariation}
TV(u^\alpha(\cdot,t))=TV(u^\alpha(\cdot,0)) = TV(u_0(\cdot))= A_2,\\
\label{equicontinuous} \int_\mathbb{R} |u^\alpha(x,t)-u^\alpha(x,s)|
dx \leq A_3|t-s|,
\end{eqnarray}
where $A_1$, $A_2$, and $A_3$ are independent of $\alpha$ and
$TV(f(\cdot))$ can be defined for a smooth function $f$ as
\begin{equation}
TV(f(\cdot))=\int_\mathbb{R} |f'(x)| dx.
\end{equation}
\end{lemma}

\begin{proof}
Property \ref{bounded} is verified by the existence proof in earlier
papers \cite{ Mohseni:07e, Mohseni:07s,Norgard:08b}, that
$||u^\alpha(\cdot,t)||_{L^\infty}=||u^\alpha(\cdot,0)||_{L^\infty}$.

To verify property \ref{totalvariation}, take the derivative of
\ref{IVP2a}, multiply by $sign(u_x)$ and integrate over the real
line to obtain
\begin{equation}
\frac{\p}{\p t} \int |u_x|\, dx + \int sgn(u_x) (\ubar u_x)_x\,dx=0.
\end{equation}
Break the second term into intervals where $sign(u_x)$ remains
constant. $u_x$ and $\ubar$ are continuous due to previous existence
theorems, so at the locations that $sign(u_x)$ switches signs, the
value of $u_x$ will be $0$. Thus the second term is zero and we
obtain the result
\begin{equation}
||u_x(\cdot,t)||_{L^1}=||u_x(\cdot,0)||_{L^1},
\end{equation}
and thus Property \ref{totalvariation} is established.

Property \ref{equicontinuous} can be proved by the following
estimate:

\begin{align*}
\int_\mathbb{R} |u^\alpha(x,t)-u^\alpha(x,s)| dx &\leq \int_\mathbb{R} \int_s^t |u_t^\alpha| \, dt \, dx\\
&=\int_\mathbb{R} \int_s^t |\ubar^\alpha u^\alpha_x| \, dt \, dx\\
&= \int_s^t \int_\mathbb{R} |\ubar^\alpha u^\alpha_x| \, dx \, dt\\
&\leq ||\ubar^\alpha||_{L^\infty} \int_s^t ||u^\alpha_x||_{L^1} dt\\
&\leq A_1 A_2 |t-s|
\end{align*}

\end{proof}

From Bressan \cite{BressanA:00a} and Serre \cite{SerreD:99a} we know
that properties (\ref{bounded}), (\ref{totalvariation}), and
(\ref{equicontinuous}) are enough to guarantee that a subsequence of
$u^\alpha$ converges to a function $u$ in $L^1_\text{loc}$.
Furthermore, $u$ shares the same infinity norm bound as that
established in (\ref{bounded}), the same total variation bound as
that in (\ref{totalvariation}).

\subsection{Convergence to a Weak Solution}\label{intro}

To begin we look at a specific subset of filters.  The filters we
examine are the functions whose Fourier transforms can be written as
$$\hat{g}(k)=\frac{1}{1+\sum_{j=1}^n C_j k^{2j}} \qquad \mbox{with } n<\infty, C_j \ge 0 \,\, C_n \ne 0.$$
Noting that $\hat{g} \hat{u}=\hat{\ubar}$  we can see that
$$\hat{u}=\left(1+\sum_{j=1}^n C_j  k^{2j}\right) \hat{\ubar}$$
and
$$u=\left(1+\sum_{j=0}^n (-1)^i C_j \frac{\p^{2j}}{\p x^{2j}}\right) \ubar$$
We will refer to a filter of this form as satisfying condition A.
This class of filters includes the Helmholtz filter, which has been
of previous interest in turbulence modeling.

Clearly $g(x)$ and its derivatives up to $g^{(2n-2)}(x)$ are well
defined and bounded as $\frac{(ik)^{2n-2}}{1+\sum_{j=1}^n C_j
k^{2j}}$ is absolutely integrable.

If $u$ and its derivative $u_x$ are absolutely integrable, then for
a $g$ satisfying condition A, the convolution
$$\frac{\p ^j}{\p x ^j} \ubar= \frac{\p ^{j-1}}{\p x ^{j-1}} g^\alpha \ast u_x$$
is well defined.  Furthermore, by Young's inequality
$$\left| \left | \frac{\p ^j}{\p x ^j} \ubar \right| \right|_\infty \leq \left| \left| \frac{\p ^{j-1}}{\p x ^{j-1}} g^\alpha  \right| \right|_\infty \left| \left| u_x \right| \right| _1=\frac{1}{\alpha^{j}} \left| \left| g^{(j-1)} \right| \right| _\infty \left| \left| u_x \right| \right|_1$$.

Thus there exists a constant $A_4$ such that
$$\left| \left | \frac{\p ^j}{\p x ^j} \ubar \right| \right|_\infty < \frac{1}{\alpha^j} A_4 \qquad \mbox{for j} \leq 2n-1 $$
This criteria is used in the following lemma.

\begin{lemma}
\label{manyderivativelimit} Let $u^\alpha$ be a sequence of
functions that satisfy the following conditions.

\begin{subequations}
\begin{eqnarray}
u^\alpha,\ubar^\alpha <A_1& \\
\int |u^\alpha_x| dx, \int |\ubar^\alpha_x| dx < A_2&\\
\left| \left | \frac{\p ^j}{\p x ^j} \ubar^\alpha \right|
\right|_\infty < \frac{1}{\alpha^j}A_4& \qquad \mbox{for j} \leq
2n-1
\end{eqnarray}
\end{subequations}
 Let $f \in C^\infty$ be compactly supported on $\mathbb{R}$.  Then as $\alpha \to 0$ the quantity
\begin{equation}
\label{alphaterm} \alpha^{2n} \int_{-\infty}^{\infty}  \left(
\frac{\p ^{2n}}{\p x ^{2n}} \ubar^\alpha \right) \, \ubar^\alpha_x f
\,  dx
\end{equation}
limits to 0.

\end{lemma}

\begin{proof}  For convenience the $u^\alpha$ shall be denoted $u$. Integrate Equation (\ref{alphaterm}) by parts to obtain
\begin{equation}
\alpha^{2n} \int   \ubar^{(2n)}  \, \ubar_x  f \,  dx =\\
\alpha^{2n} \int  \ubar^{(n)} \,  \frac{\p ^{n}}{\p x ^{n}} (\ubar_x
f) \,  dx.
\end{equation}

Use product rule to expand $\left( \frac{\p ^{n}}{\p x ^{n}} \ubar_x
f  \right)$

\begin{equation}
\alpha^{2n} \int  \ubar^{(n)} \, \sum_{i=0}^n \binom{n}{i}
\ubar^{(1+i)} f^{(n-i)}  \,  dx.
\end{equation}

Take the absolute value, separate the last two terms of the binomial
expansion, and apply the triangle inequality.

\begin{eqnarray}
\leq & \left| \alpha^{2n} \int  \ubar^{(n)} \ubar^{(n+1)} f  \,dx
\right|\\
&  + \left| \alpha^{2n} n \int  \ubar^{(n)}  \ubar^{(n)}
f^{(1)} \,dx \right|\\
&  + \left| \alpha^{2n} \int  \ubar^{(n)} \, \sum_{i=0}^{n-2}
\binom{n}{i} \ubar^{(1+i)} f^{(n-i)}  \,  dx \right|
\end{eqnarray}

Begin by bounding the third term,
\begin{equation}
\underbrace{\left| \alpha^{2n} \int  \ubar^{(n)} \, \sum_{i=0}^{n-2}
\binom{n}{i} \ubar^{(1+i)} f^{(n-i)}  \,  dx \right|}_{\text{3rd
term}} \leq \alpha^{2n} \sum_{i=0}^{n-2} \binom{n}{i}
||\ubar^{(n)}||_\infty ||\ubar^{(1+i)}||_\infty   ||f^{(n-i)}||_1
\end{equation}

By applying the bound on $||\ubar^{(i)}||_\infty$,
\begin{equation}
\text{3rd term} \leq \sum_{i=0}^{n-2}  \binom{n}{i} \alpha^{n-i-1}
A_4||f^{(n-i)}||_1,
\end{equation}
which limits to $0$ as $\alpha \to 0$.

Next deal with the second term.
\begin{align}
\underbrace{\left| \alpha^{2n} n \int  \ubar^{(n)}  \ubar^{(n)}  f^{(1)}  \,dx \right|}_{\text{2nd term}} =& \\
=& \left| \alpha^{2n} n \int  \ubar^{(1)}  \frac{\p ^{n-1}}{\p x ^{n-1}} ( \ubar^{(n)}  f^{(1)} )  \,dx \right|\\
=& \left| \alpha^{2n} n \int  \ubar^{(1)}  \sum_{i=0}^{n-1} \binom{n-1}{i} \ubar^{(n+i)}  f^{(n-i)} )  \,dx \right|\\
\leq & \alpha^{2n} n \sum_{i=0}^{n-1} \binom{n-1}{i}
||\ubar^{(1)}||_1  ||\ubar^{(n+i)}||_\infty  ||f^{(n-i)}||_\infty
\end{align}

Again, apply the bound on $||\ubar^{(i)}||_\infty$ to get
\begin{equation}
\text{2nd term} \leq  n \sum_{i=0}^{n-1} \alpha^{n-i} \binom{n-1}{i}
A_4 ||\ubar^{(1)}||_1   ||f^{(n-i)}||_\infty.
\end{equation}
Since $f$ and all its derivatives are bounded and $||\ubar^{(1)}||_1
< A_2 $ the second term also limits to zero.

Now for the first term .
\begin{eqnarray}
\left| \alpha^{2n} \int  \ubar^{(n)} \ubar^{(n+1)} f  \,dx \right| =
&
\left| \alpha^{2n} \int  \frac{1}{2} \frac{\p}{\p x} ( \ubar^{(n)} )^2  f  \,dx \right|\\
=& \left| \frac{\alpha^{2n}}{2}  \int   \ubar^{(n)}  \ubar^{(n)}
f^{(1)}  \,dx \right|
\end{eqnarray}

This differs from the second term only by a constant,  so it must
limit to 0 as $\alpha \to 0$.

Thus we obtain the result
\begin{equation}
\lim _{\alpha \to 0} \alpha^{2n} \int_{-\infty}^{\infty}  \left(
\frac{\p ^{2n}}{\p x ^{2n}} \ubar \right) \, \ubar_x  f \,  dx =0.
\end{equation}
\end{proof}

The last piece needed is taken from Duoandikoetxea
\cite{DuoandikoetxeaJ:99a}.  The following lemma is a restatement of
Duoandikoetxea' Theorem 2.1 from page 25.

\begin{lemma}
\label{convolutionlemma} Let $g$ be an integrable function on
$\mathbb{R}$ such that $\int g =1$.  Define
$g^\alpha=\frac{1}{\alpha} g(\frac{x}{\alpha})$.  Then
$$\lim_{\alpha \to 0}  ||g^\alpha \ast f - f||_p = 0$$  if $f \in L^p, 1\leq p < \infty$ and uniformly (i.e. when $p=\infty$) if $f \in C_0(\mathbb{R})$.
\end{lemma}

With lemmas \ref{convolutionlemma} and \ref{manyderivativelimit} we
can now prove the following theorem regarding convergence to weak
solutions.

\begin{theorem}
For any $g$ satisfying condition A, the solutions $u^\alpha$ to the
CFB equations converge to a weak solution of the inviscid Burgers
equation.
\end{theorem}

\begin{proof}
It was already shown that $u^\alpha$ converges to a function $u$. To
show this is a weak solution of the inviscid Burgers equation, we
need to prove that for any test function $f \in C^\infty$ that has
compact support on $\mathbb{R} \times [0,T]$ that
\begin{equation}
\label{weaksolution} \int_0^T \int_\mathbb{R} u f_t +\frac{1}{2}u^2
f_x \,dx\,dt=0.
\end{equation}

Begin by rewriting Equation \ref{IVP2a} as
\begin{equation}
\label{rewrite}
u^\alpha_t+\left(\frac{1}{2}(\ubar^\alpha)^2\right)_x=(\ubar_x^\alpha-u_x^\alpha)\ubar^\alpha.
\end{equation}

Multiply by the test function $f$ and integrate over $\mathbb{R}
\times [0,T]$
\begin{equation}
\int_0^T  \int_ \mathbb{R} u^\alpha_t f
+\left(\frac{1}{2}(\ubar^\alpha)^2\right)_x f \, dx dt =\int_0^T
\int_ \mathbb{R}  (\ubar_x^\alpha-u_x^\alpha)\ubar^\alpha f \, dx dt
\end{equation}

Integrate by parts.
\begin{eqnarray}
 \int_0^T  \int_ \mathbb{R} u^\alpha f_t
+\left(\frac{1}{2}(\ubar^\alpha)^2\right) f_x \, dx dt =\int_0^T
\int_ \mathbb{R} (\ubar^\alpha-u^\alpha)\ubar^\alpha f_x \, dx dt \nonumber \\
\label{fullweakform} + \int_0^T  \int_ \mathbb{R}
(\ubar^\alpha-u^\alpha)\ubar_x^\alpha f \, dx dt
\end{eqnarray}

Taking the limit as $\alpha \to 0$ of the left side,  you get
$$\int_0^T  \int_ \mathbb{R} u f_t +\left(\frac{1}{2}(u)^2\right) f_x \, dx dt$$
Clearly then if the right hand side limits to zero, we have $u$ is a
weak solution to Burgers equation.

Begin with the first term on the right hand side of Equation
\ref{fullweakform}. The first term can be shown to limit to zero by
noting that $||u^\alpha||_{\infty}$ has a uniform bound of $A_1$,
and that since $f \in C^\infty$ with compact support, there exists
an $F \in \mathbb{R}^+$ such that $||f||_{\infty} \leq  F$ and
$||f_x||_{\infty} \leq  F$. Additionally let $f$ be supported on the
compact set $\Omega.$ This leads to the bound
\begin{equation}
\int_0^T \int_\mathbb{R} (u^\alpha-\ubar^\alpha)u^\alpha f_x\,dx\,dt
\leq F\, A_1\, T ||u^\alpha-\ubar^\alpha||_{L^1(\Omega)}.
\end{equation}

Take the limit of $||u^\alpha-\ubar^\alpha||_{L^1(\Omega)}$.  Break
apart the norm with the triangle inequality to get
\begin{eqnarray*}
\lim_{\alpha \to 0} ||u^\alpha-u^\alpha \ast g^\alpha|| & \leq &\lim_{\alpha \to 0} ||u^\alpha-u|| +||u-u \ast g^\alpha|| +|| g^\alpha \ast (u-u^\alpha)||\\
& \leq &\lim_{\alpha \to 0} ||u^\alpha-u|| +||u-u \ast g^\alpha||
+|| g^\alpha|| || u-u^\alpha||,
\end{eqnarray*}
where the norms are all  $||\cdot||_{L^1(\Omega)}$.  The first and
third term limit to zero as $u^\alpha$ converges to $u$ in
$L^1_{loc}$.  The second term limits to zero by lemma
\ref{convolutionlemma}.

Now deal with the second term from equation (\ref{fullweakform}).
Since $g$ satisfies condition A,
$$(\ubar^\alpha-u^\alpha)=\sum_{j=1}^n C_i \alpha^{2j} \frac{\p^{2j}}{\p x ^{2j}} \ubar^{\alpha}$$ the second term can be rewritten as

$$\sum_{j=1}^n C_i \int_0^T  \int_ \mathbb{R} \alpha^{2j} \frac{\p^{2j}}{\p x ^{2j}} \ubar^{\alpha} \ubar_x^\alpha f \, dx dt$$

By Lemma \ref{manyderivativelimit} every term in the sum limits to
zero.  Hence the sum limits to zero.

Therefore the limit as $\alpha \to 0$ of Equation
(\ref{fullweakform}) becomes
\begin{equation}
\int_0^T \int_\mathbb{R} u f_t +\frac{1}{2}u^2 f_x \,dx\,dt=0.
\end{equation}
proving $u$ is a weak solution of the inviscid Burgers equation.
\end{proof}

\section{Convergence to the Entropy Solution}\label{entropysolutionsection}

In this section we will first examine some of the properties of
non-entropic solutions, that is solutions that are a weak solution
to the inviscid Burgers Equation, but do not satisfy the entropy
condition.  By examining these properties, it will be shown that the
solutions to the CFB equation lack certain properties found in all
non-entropic solution.  Thus it will be shown that the solutions to
the CFB equations converge to the entropy solution of the inviscid
Burgers Equation.

This examination will be limited to a class of initial conditions.
Specifically, we intend to examine initial conditions that are
continuously differentiable, and are bell shaped, i.e. have an
interval where the functions are increasing, followed by an interval
where the functions are decreasing.  Functions that satisfy this
condition will be referred to as satisfying condition B.  It is for
these functions as initial conditions that we will prove convergence
to the entropy solution.

\begin{definition}\label{conditionB}
Let $u(x) \in C^1(\mathbb{R})$ and  $u_x \ge 0$ over $(-\infty, p)$
and $u_x \le 0$ over $(p, -\infty)$ for some $p$.  Additionally let
$u(x)$ have finite limits as $x \to \pm \infty$  Then $u(x)$ is said
to have satisfied condition B.
\end{definition}

\subsection{Non-entropic Weak Solutions}

There are three classic types of entropy violating weak solutions to
the inviscid Burgers equation. This subsection shows examples of
each type.  The first is when you start with an increasing shock in
the initial conditions and then that shock remains, propagating at
the speed dictated by the Rankine-Hugoniot jump conditions. An
example of this is
\begin{equation*}
u(x,t)=
\begin{cases}
0 & \text{if $x<\frac{1}{2}t$,}\\
1 &\text{if $  \frac{1}{2}t \le x$},
\end{cases}
\end{equation*}
taken from Lax \cite{LaxPD:73a} and is illustrated in figure
\ref{example1}.

\begin{figure}[!ht]
\begin{center}
\begin{minipage}{0.48\linewidth} \begin{center}
  \includegraphics[width=.9\linewidth]{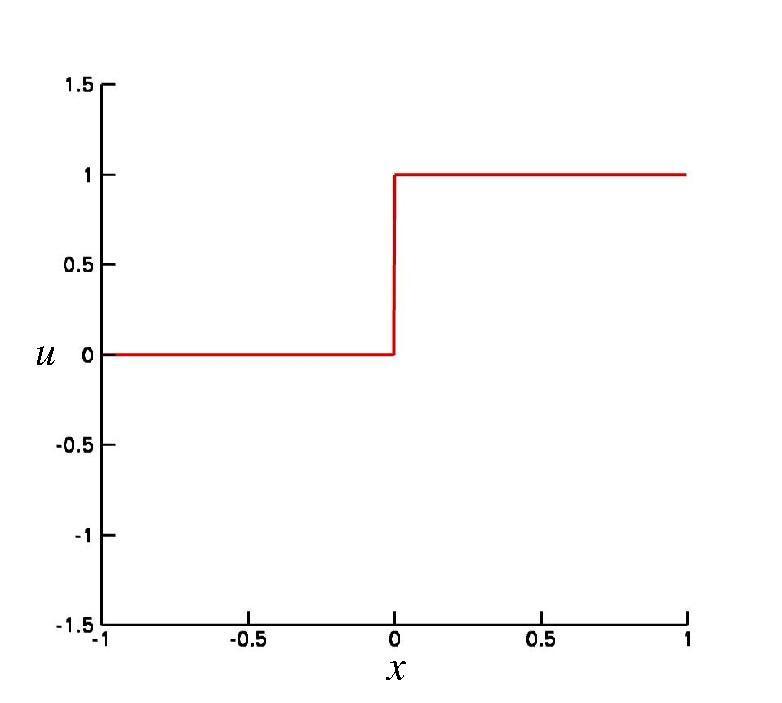}
\end{center} \end{minipage}
\begin{minipage}{0.48\linewidth} \begin{center}
  \includegraphics[width=.9\linewidth]{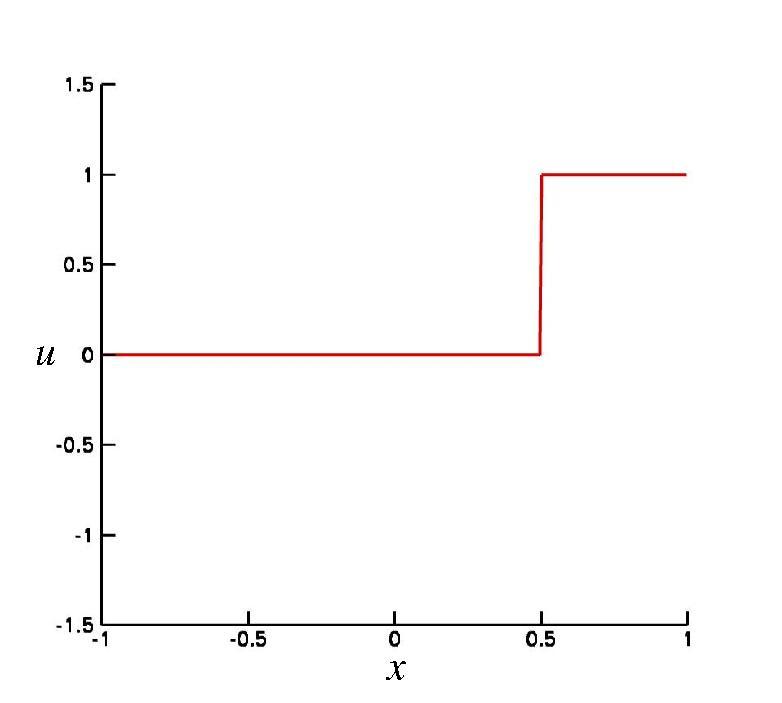}
\end{center} \end{minipage}\\
\begin{minipage}{0.48\linewidth}\begin{center} (a) \end{center} \end{minipage}
\begin{minipage}{0.48\linewidth}\begin{center} (b) \end{center}
\end{minipage}\vspace{-2mm}
\caption{Here a discontinuity is introduced in the initial
conditions and remains.  The shock must travel at the speed dictated
by the Rankine-Hugoniot jump conditions to be a weak solution.}
 \label{example1}
\end{center}
\end{figure}

The second case is when a shock already exists and then splits into
multiple shocks, one of which is an entropy violating shock. All the
shocks move with the speed dictated by the Rankine-Hugoniot
conditions. For $a \ge 1$ the following is a weak solution to
Burgers Equation.  This example was taken from Oleinik
\cite{OleinikOA:57a} and is illustrated in figure \ref{example2}
\begin{equation*}
u(x,t)=
\begin{cases}
1 & \text{if $x<\frac{1-a}{2}t$,}\\
-a &\text{if $\frac{1-a}{2}t \le x < 0$}\\
a &\text{if $0 \le x < \frac{a-1}{2}t$}\\
-1 &\text{if $  \frac{a-1}{2}t \le x$}
\end{cases}
\end{equation*}

\begin{figure}[!ht]
\begin{center}
\begin{minipage}{0.48\linewidth} \begin{center}
  \includegraphics[width=.9\linewidth]{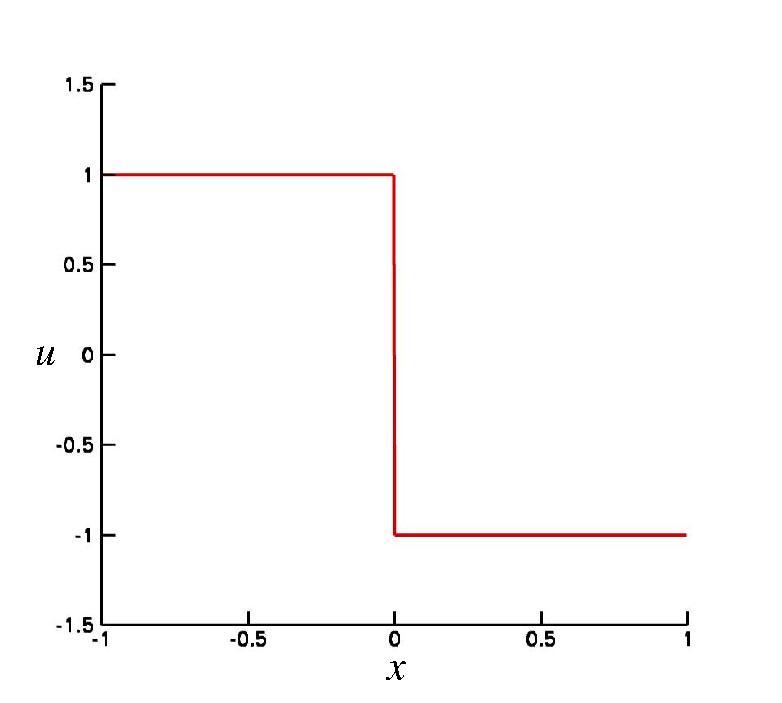}
\end{center} \end{minipage}
\begin{minipage}{0.48\linewidth} \begin{center}
  \includegraphics[width=.9\linewidth]{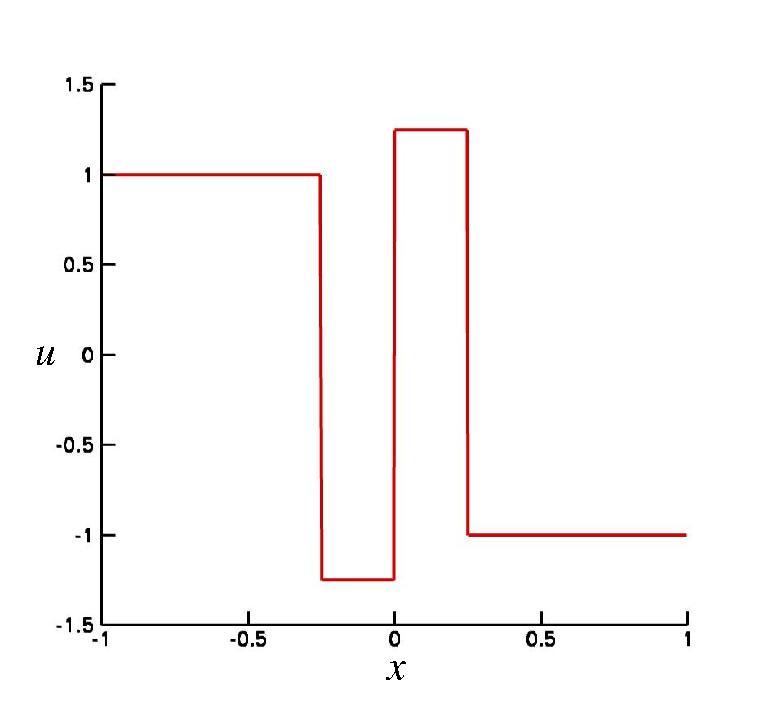}
\end{center} \end{minipage}\\
\begin{minipage}{0.48\linewidth}\begin{center} (a) \end{center} \end{minipage}
\begin{minipage}{0.48\linewidth}\begin{center} (b) \end{center}
\end{minipage}\vspace{-2mm}
\caption{A shock can split into multiple shocks and still remain a
weak solution.  (a)  The initial conditions.  (b)  The solution
after the shock splitting has occurred. } \label{example2}
\end{center}
\end{figure}

Another example is spontaneous shock formation with shocks forming
out of a continuous interval.  For $a>0$ the following is a weak
solution to Burgers Equation.  This example was taken from Serre
\cite{SerreD:99a} and is illustrated in figure \ref{example3}.
\begin{equation*}\label{spontanteousexample}
u(x,t)=
\begin{cases}
0 & \text{if $x<\frac{-a}{2}t$,}\\
-a &\text{if $\frac{-a}{2}t \le x < 0$}\\
a &\text{if $0 \le x < \frac{a}{2}t$}\\
0 &\text{if $  \frac{a}{2}t \le x$}
\end{cases}
\end{equation*}

\begin{figure}[!ht]
\begin{center}
\begin{minipage}{0.48\linewidth} \begin{center}
  \includegraphics[width=.9\linewidth]{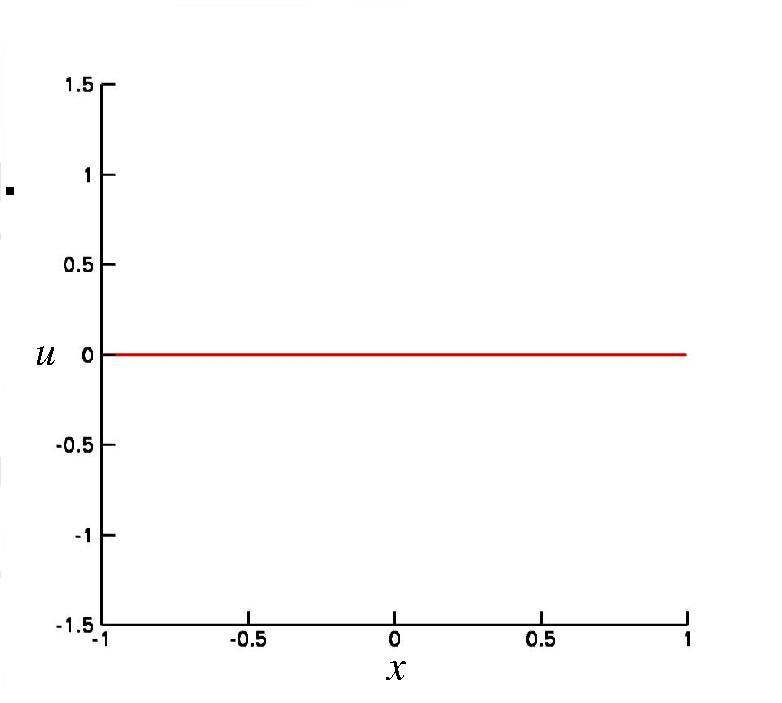}
\end{center} \end{minipage}
\begin{minipage}{0.48\linewidth} \begin{center}
  \includegraphics[width=.9\linewidth]{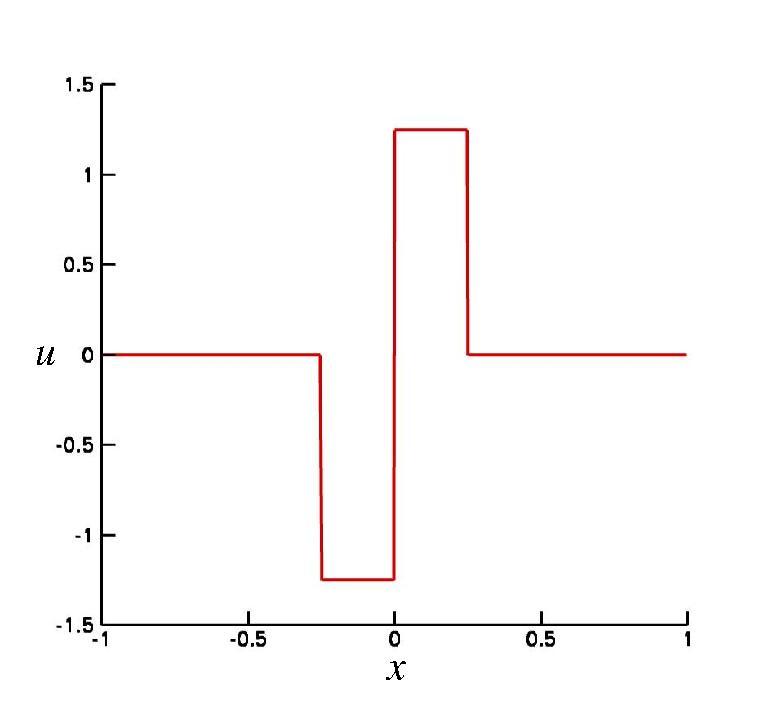}
\end{center} \end{minipage}\\
\begin{minipage}{0.48\linewidth}\begin{center} (a) \end{center} \end{minipage}
\begin{minipage}{0.48\linewidth}\begin{center} (b) \end{center}
\end{minipage}\vspace{-2mm}
\caption{A shock can form from a continuous interval.  (a)  The
initial conditions.  (b)  The solution after spontaneous shock
formation has occurred. } \label{example3}
\end{center}
\end{figure}

In the next section it is shown that these three cases exemplify the
only type of entropy violating behavior possible.

\subsection{Decreasing Slope along characteristics}\label{decreasingslopesection}

By examining the inviscid Burgers Equation, it is possible to see
that a non-entropic solution cannot form through the steepening of
the solution.  With this information we can then limit the ways a
non-entropic solution can come into being. Begin with the inviscid
Burgers Equation
\begin{equation}
 u_t+uu_x=0.
\end{equation}

In section \ref{previouslyestablishedfacts} it was seen that along
the characteristics
\begin{equation}
\xi=x_0+u_0(\xi)t
\end{equation}
that the value of $u$ remains constant.  This is true until
characteristics crossed, at which point a shock is formed.

Here a similar approach is taken, but on the derivative of the
inviscid Burgers Equation.  Differentiate the inviscid Burgers
Equation to get

\begin{equation}
 \frac{d}{dt} (u_x)+ u\frac{d}{dx}(u_x)=-(u_x)^2
\end{equation}

Now if you examine this equation you find that along the same
characteristics $\xi=x_0+u_0(\xi)t$ that the quantity $u_x$ is
governed by
\begin{equation}
 \frac{d}{dt} u_x=-(u_x)^2.
\end{equation}
Thus for piecewise differentiable solutions, $u_x$ is always
decreasing along characteristics and an increasing shock cannot form
from the steepening of the solution.

Now consider a solution that begins with initial conditions
satisfying condition B.  That solution is a continuously
differentiable solution to the inviscid Burgers equation and is thus
an entropy solution. It will remain an entropy solution until an
increasing jump is formed.  An entropy solution for initial
conditions satisfying condition B will be piecewise continuous and
thus from above will not steepen into an increasing shock.  From
this we conclude that an increasing shock can only occur if it
exists in the initial conditions or must form instantaneously as it
cannot form from the steepening of the solution.  It can either form
at existing points of discontinuity or form at points of continuity,
which this paper refers to as shock splitting and spontaneous shock
formation respectively.

\subsection{Entropy violating solutions are not reparameterizations of initial conditions}

In section \ref{previouslyestablishedfacts} it was established that
the entropy solution of the inviscid Burgers equation is a
reparameterization of initial conditions when the initial conditions
are continuous.  This subsection shows that a non-entropic solution
cannot be both a weak solution and a reparameterization of initial
conditions satisfying condition B.

We first begin be examining some consequences of being both a weak
solution and a reparameterization of initial conditions satisfying
condition B. Then we assume that there is a non-entropic solution
that is both a weak solution and a reparameterization and show that
this is a contradiction.

If a function is a reparameterization of initial conditions
satisfying condition B, it is easy to see that the
reparameterization will have one interval where it is increasing
followed by an interval where it is decreasing.  It is also clearly
bounded. However, it need not be continuous.  As a direct
consequence of the Monotone Convergence Theorem for sequences, every
point on the reparameterization will have a well defined left and
right sided limit. Since the left and right sided limits are well
defined, the only type of discontinuity allowed is a jump
discontinuity.  If a more rigorous explanation is desired we refer
the reader to Section 5.7 in Davidson and Donsig
\cite{DonsigAP:02a}.

Additionally any function satisfying condition B will have bounded
variation.  Thus any function that is a reparameterization will have
variation bounded by the original function's variation.  Thus if a
solution is a reparameterization of initial conditions initial
conditions satisfying condition B then it is of bounded variation.

From Theorem 1.8.1 on page 21 and page 52 in Dafermos
\cite{DafermosCM:99a} we know that a function $u$ that is of class
$BV_{loc}$ and is a weak solution will satisfy the Rankine-Hugoniot
jump conditions at every jump discontinuity. This means that if
$\chi$ is the location of a discontinuity then
\begin{equation}\label{rankinehugoniotjumpcondition}
\frac{d}{dt}\chi= \frac{u(\chi^-,t)+u(\chi^+,t)}{2}.
\end{equation}
Thus if the solution is a weak solution and a reparameterization of
initial conditions satisfying condition B, all of its
discontinuities must be jump discontinuities satisfying the
Rankine-Hugoniot jump conditions.

To show that a function is not a reparameterization of initial
conditions satisfying condition B, it is sufficient to find three
points $x_1<x_2<x_3$ such that $u(x_1)>u(x_2)$ and $u(x_2)< u(x_3)$.
Essentially a function satisfying condition B is bell shaped and
finding these three points finds an upsidedown bell, which cannot
happen in a reparameterization.  This is precisely the method used
to show that a non-entropic solution cannot be both a
reparameterization of the initial conditions.

Since we are considering only initial conditions satisfying
condition B, we are only beginning with continuous initial
conditions. Thus from section \ref{decreasingslopesection} the only
possibility of having a non-entropic solution is through spontaneous
shock formation or by shock splitting.  It will be shown that if
either of these occur, then the non-entropic solution fails to be a
reparameterization of the initial conditions.

The following lemma is used later on when dealing with spontaneous
shock formation and shock splitting.  Because a non-entropic
solution must still be a weak solution, spontaneous shock formation
and shock splitting must behave is certain ways. The inviscid
Burgers equation can be considered as a conservation law of
wavemass, $\int u$.  Lemma \ref{thatonelemma} uses this fact to
place restrictions on how spontaneous shock formation and shock
splitting can occur.

Lemma \ref{thatonelemma} addresses the area between the leftmost and
rightmost shock, when a spontaneous shock formation or shock
splitting occur.  Essentially it says that if the area between the
shocks has a higher value than the value on the outside of the
shocks, then wavemass has been created, and it is no longer a weak
solution to the inviscid Burgers equation.  Figure
\ref{wavemasscreated}, shows illustrations of this.

\begin{figure}[!ht]
\begin{center}
\begin{minipage}{0.48\linewidth} \begin{center}
  \includegraphics[width=.9\linewidth]{FIGS/discpropogate1.jpg}
\end{center} \end{minipage}
\begin{minipage}{0.48\linewidth} \begin{center}
  \includegraphics[width=.9\linewidth]{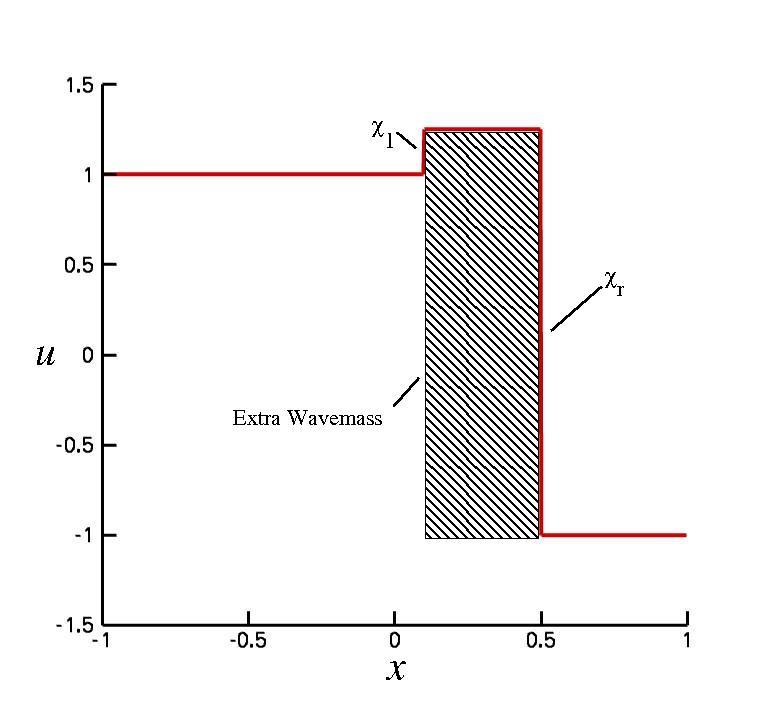}
\end{center} \end{minipage}\\
\begin{minipage}{0.48\linewidth}\begin{center} (a) \end{center} \end{minipage}
\begin{minipage}{0.48\linewidth}\begin{center} (b) \end{center}
\end{minipage}\vspace{-2mm}
\caption{Wavemass is created in shock splitting or spontaneous shock
formation if the middle area is greater than the outer area's values
at the shocks. (a)The conditions of $u(x,t)$ before shock splitting.
(b)If a shock splits and the middle area has higher values than its
surroundings, extra wavemass has been created.  Here $\chi_l$ and
$\chi_r$ are the left most and right most shocks as used in
lemma\ref{thatonelemma}.  The extra wavemass is indicated. }
\label{wavemasscreated}
\end{center}
\end{figure}

The lemma proves that if the area between the leftmost and rightmost
shock has values greater than those on its borders, then $u(x,t)$
cannot be a weak solution of the inviscid Burgers equation and a
reparameterization of initial conditions.  This is a proof by
contradiction, so we assume that $u(x,t)$ is a weak solution and a
reparameterization of initial conditions satisfying condition B,
which places several constraints on it. Such a weak solution to the
inviscid Burgers equation must satisfy the Rankine-Hugoniot jump
conditions, Equation (\ref{rankinehugoniotjumpcondition}).

Additionally, weak solutions must satisfy the integral form of the
conservation law
\begin{equation}
\left. \int_{g}^{h} u \, dx  \right|_{t_1}^{t_2}=\left.
\int_{t_1}^{t_2} \frac{-u^2}{2} \right|_{g}^{h} \, dt,
\end{equation}
or if $g$ and $h$ are moving boundaries
\begin{equation}\label{weaksolutiondefinitionwithmovingboundaries}
\left. \int_{g(t)}^{h(t)} u \, dx  \right|_{t_1}^{t_2}=\left.
\int_{t_1}^{t_2} \frac{-u^2}{2} \right|_{g(t)}^{h(t)}  +
\left(\frac{d}{dt} h(t)\right)u(h(t),t) - \left(\frac{d}{dt}
g(t)\right)u(g(t),t)  \, dt.
\end{equation}
The second definition, Equation
\ref{weaksolutiondefinitionwithmovingboundaries}, is used in the
following lemma.

\begin{lemma}\label{thatonelemma}
Assume that $u(x,t)$ takes the form
\begin{equation}
u(x,t)=
\begin{cases}
a(x,t) & \text{if $          x <\chi_l(t)$,}\\
b(x,t) & \text{if $\chi_l(t) \le x < \chi_r(t)$}\\
c(x,t) & \text{if $\chi_r(t) \le x $},
\end{cases}
\end{equation}
where $\chi_l(t)$ and $\chi_r(t)$ are locations of discontinuities
and $\chi_l(t_1)=\chi_r(t_1)=x^*$.  At time $t_1$ let $a(x^{*-},t_1)
\geq c(x^{*+},t_1).$  If for some period of time after $t_1$ and all
$x\in(\chi_l,\chi_r)$, $b(x,t)>a(\chi_l(t)^-,t)$  and
$b(x,t)>c(\chi_r(t)^+,t)$, then $u(x,t)$ cannot be a weak solution
of the inviscid Burgers equation.
\end{lemma}

\begin{proof}
Begin by assuming that $u(x,t)$ is a weak solution of the inviscid
Burgers equation and thus must satisfy Equation
\ref{weaksolutiondefinitionwithmovingboundaries} for any $g(t)$ and
$h(t)$.  We will start by considering the left hand side of Equation
\ref{weaksolutiondefinitionwithmovingboundaries} with selected
moving boundaries and show that it is strictly greater than the
right hand side, proving that $u(x,t)$ cannot be a weak solution by
contradiction. The moving boundaries will be defined by the
positions of the leftmost and rightmost shock.

With the moving boundaries established, begin with the left hand
side of Equation (\ref{weaksolutiondefinitionwithmovingboundaries})
for the given boundaries. By putting a bound on the integrand, we
transform the spatial integral into a temporal integral.
\begin{align}
 \underbrace{\left. \int_{\chi_l(t)}^{\chi_r(t)} u\, dx  \right |_{t_1}^{t_2}}_{LHS} &\geq \min_{x \in (\chi_l \chi_r)} b(x,t_2) \left( \chi_r-\chi_l \right)\\
 &=\min_{x \in (\chi_l \chi_r)} b(x,t_2) \int_{t_1}^{t_2} \frac{\p}{\p t} \chi_r-\frac{\p}{\p t}\chi_l \,\, dt.
\end{align}

Now manipulate the equation to begin resembling the right hand side
of Equation \ref{weaksolutiondefinitionwithmovingboundaries}.
\begin{align*}
LHS & \geq  \int_{t_1}^{t_2} \min_{x \in (\chi_l \chi_r)} b(x,t_2)\left( \frac{\p}{\p t} \chi_r- \frac{\p}{\p t}\chi_l \right)  dt \\
&= \int_{t_1}^{t_2} \left( \min_{x \in (\chi_l \chi_r)} b(x,t_2) -c(\chi_r,t) \right) \left(\frac{\p}{\p t} \chi_r- \frac{\p}{\p t}\chi_l \right ) \\
&         \qquad \,        +\left(a(\chi_l,t)-c(\chi_r,t)\right) \frac{\p}{\p t}\chi_l  \\
&         \qquad \,        + c(\chi_r,t)\frac{\p}{\p t} \chi_r
-a(\chi_l,t)\frac{\p}{\p t}\chi_l  \,\, dt
\end{align*}
Substitute in the speed of $\chi_l$ dictated by the Rankine-Hugoniot
jump conditions.
\begin{align*}
LHS & \geq \int_{t_1}^{t_2} \left( \min_{x \in (\chi_l \chi_r)} b(x,t_2) -c(\chi_r,t) \right) \left(\frac{\p}{\p t} \chi_r- \frac{\p}{\p t}\chi_l \right ) \\
&         \qquad \,        +\left(a(\chi_l,t)-c(\chi_r,t)\right) \left(\frac{b(\chi_l,t)+a(\chi_l,t)}{2}\right)  \\
&         \qquad \,        + c(\chi_r,t)\frac{\p}{\p t} \chi_r -a(\chi_l,t)\frac{\p}{\p t}\chi_l   \,\, dt\\
&=   \int_{t_1}^{t_2} \underbrace{\left( \min_{x \in (\chi_l \chi_r)} b(x,t_2) -c(\chi_r,t) \right) \left(\frac{\p}{\p t} \chi_r- \frac{\p}{\p t}\chi_l \right )}_{\text{L}} \\
&         \qquad \,        +\underbrace{\left(a(\chi_l,t)-c(\chi_r,t)\right) \left(\frac{b(\chi_l,t)-c(\chi_r,t)}{2}\right)}_{\text{M}}  \\
&         \qquad \,        +\left(a(\chi_l,t)-c(\chi_r,t)\right)
\left(\frac{c(\chi_r,t)+a(\chi_l,t)}{2}\right)+
c(\chi_r,t)\frac{\p}{\p t} \chi_r -a(\chi_l,t)\frac{\p}{\p t}\chi_l
\,\, dt
\end{align*}
Consider term L.  The value of $b(x,t)$ for all $x$ and some period
of time after $t_1$ was designated to be higher than  $c(\chi_r,t)$.
Additionally, for at least a short period of time $\frac{\p}{\p t}
\chi_r > \frac{\p}{\p t}\chi_l$, otherwise the interval
$(\chi_l,\chi_r)$ cannot have a nonzero measure.  Thus for values
$t_2$ close to $t_1$, term L is strictly positive.

Now consider term M.  Again the value of $b(x,t)$ for all $x$ and
some period of time after $t_1$ was designated to be higher than
$c(\chi_r,t)$.  It was designated that at time $t_1$, $a(x^{*-},t_1)
\geq c(x^{*+},t_1).$  If $a(x^{*-},t_1) > c(x^{*+},t_1)$, then for
values $t_2$ close to $t_1$ term M is strictly positive.  If
$a(x^{*-},t_1) = c(x^{*+},t_1)$, then by choosing $t_2$ close to
$t_1$, term M can be made arbitrarily small.

As $t_2$ approaches $t_1$, term L is approaching a strictly positive
number, and term M is approaching a non-negative number.  Thus it is
possible to choose a $t_2$ where $\int_{t_1}^{t_2} \text{L}+\text{M}
\, dt>0$.  Using this we see that
\begin{equation}
\label{leftisbiggerthanright} LHS> \int_{t_1}^{t_2}
\frac{a(\chi_l,t)^2-c(\chi_r,t)^2}{2} + c(\chi_r,t)\frac{\p}{\p t}
\chi_r -a(\chi_l,t)\frac{\p}{\p t}\chi_l  \,\, dt.
\end{equation}
The right hand side of Equation \ref{leftisbiggerthanright} is the
right hand side of Equation
\ref{weaksolutiondefinitionwithmovingboundaries} with our chosen
boundaries.  Since with our moving boundaries the left hand side of
Equation \ref{weaksolutiondefinitionwithmovingboundaries} is
strictly greater than the right hand side, $u(x,t)$ cannot be a weak
solution. \end{proof}

This result is now used to show that if there is spontaneous shock
formation or shock splitting that $u(x,t)$ cannot be both a weak
solution and a reparameterization of initial conditions.

\subsubsection{Spontaneous shock formation}\label{spontaneousshocks}

Assume that $u(x,t)$ is the entropy solution to the inviscid Burgers
equation up to time $t_1$ where an increasing shock is formed
spontaneously at point $x^*$.  For such a discontinuity to form at
least one other discontinuity must form in response.  Thus we say
after time $t_1$,
\begin{equation}
u(x,t)=
\begin{cases}
a(x,t) & \text{if $          x <\chi_l(t)$,}\\
b(x,t) & \text{if $\chi_l(t) \le x < \chi_r(t)$}\\
c(x,t) & \text{if $\chi_r(t) \le x $},
\end{cases}
\end{equation}
where $a(x,t), b(x,t),$ and $c(x,t)$ are weak solutions to the
inviscid Burgers equation and $\chi_l(t)$ and $\chi_r(t)$ give the
position of the leftmost and rightmost discontinuities formed during
the shock splitting.  There may be more than two shocks formed as
seen in section (\ref{spontanteousexample}), but we just need to
examine the leftmost and right most.

For $u(x,t)$ to be a weak solution we note that several things must
be true.  The speed of $\chi_l(t)$ and $\chi_r(t)$ are dictated by
the Rankine-Hugoniot jump conditions to be
\begin{equation}
\frac{d}{dt}\chi_l= \frac{a(\chi_l^-)+b(\chi_l^+)}{2}  \qquad
\frac{d}{dt}\chi_r= \frac{b(\chi_r^-)+c(\chi_r^+)}{2}.
\end{equation}
For there to be spontaneous shock forming, there must be some
interval $(t_1, t_2)$, where $\frac{d}{dt}\chi_r > \frac{d}{dt}
\chi_l$.  Thus for some interval $(t_1, t_2)$,  if $a(\chi_l^-) \ge
c(\chi_r^+)$ then $b(\chi_l(t)^+,t)<b(\chi_r(t)^-,t)$.  Assume that
$t \in (t_1, t_2)$ for the remainder of the subsection.

The shocks located at $\chi_l(t)$ and $\chi_r(t)$ must either be
increasing or decreasing shocks.  We will examine each of the
possibilities and show that each leads to $u(x,t)$ not being a
reparameterization of the initial conditions.

\paragraph{Case 1}
Assume that the shock at $\chi_l(t)$ is a decreasing shock and the
shock at $\chi_r(t)$ is a decreasing shock. Then
$a(\chi_l(t)^-,t)>b(\chi_l(t)^+,t)$ and $ b(\chi_r(t)^-,t)>
c(\chi_r(t)^+,t)$.  If $b(\chi_l(t)^+,t) \geq b(\chi_r(t)^-,t)$,
then by the transitive property $a(\chi_l(t)^-,t)>c(\chi_r(t)^+,t)$
and this violates the Rankine-Hugoniot condition,  as was mentioned
above, and $u(x,t)$ is not a weak solution.  If
$b(\chi_l(t)^+,t)<b(\chi_r(t)^-,t)$, then
$a(\chi_l(t)^-,t)>b(\chi_l(t)^+,t)<b(\chi_r(t)^-,t)$, shows $u(x,t)$
is not a reparameterization of initial conditions.  See figure
\ref{case1}.

 \begin{figure}[!ht]
 \begin{center}
 \includegraphics[width=0.9\linewidth]{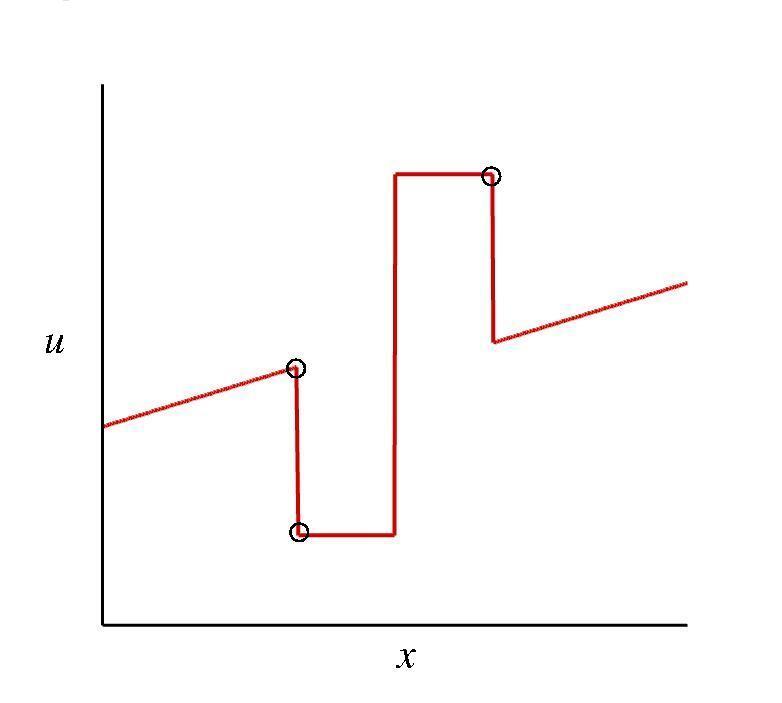}
 \caption{With two decreasing shocks, $u(x,t)$ cannot be a reparameterization.
 The circles represent the points chosen to prove that $u(x,t)$ cannot be
 a reparameterization of initial conditions.}
 \label{case1}
 \end{center}
 \end{figure}

\paragraph{Case 2}
Assume that the shock at $\chi_l(t)$ is a decreasing shock and the
shock at $\chi_r(t)$ is an increasing shock.  Then
$a(\chi_l(t)^-,t)>b(\chi_l(t)^+,t)$ and $ b(\chi_r(t)^-,t)<
c(\chi_r(t)^+,t)$.  Let $b_2=\min\left(
b(\chi_l(t)^+,t),b(\chi_r(t)^-,t) \right)$ , then
$a(\chi_l(t)^-,t)>b2<c(\chi_r(t)^+,t)$, shows $u(x,t)$ is not a
reparameterization of initial conditions. See figure \ref{case2}.

 \begin{figure}[!ht]
 \begin{center}
 \includegraphics[width=0.9\linewidth]{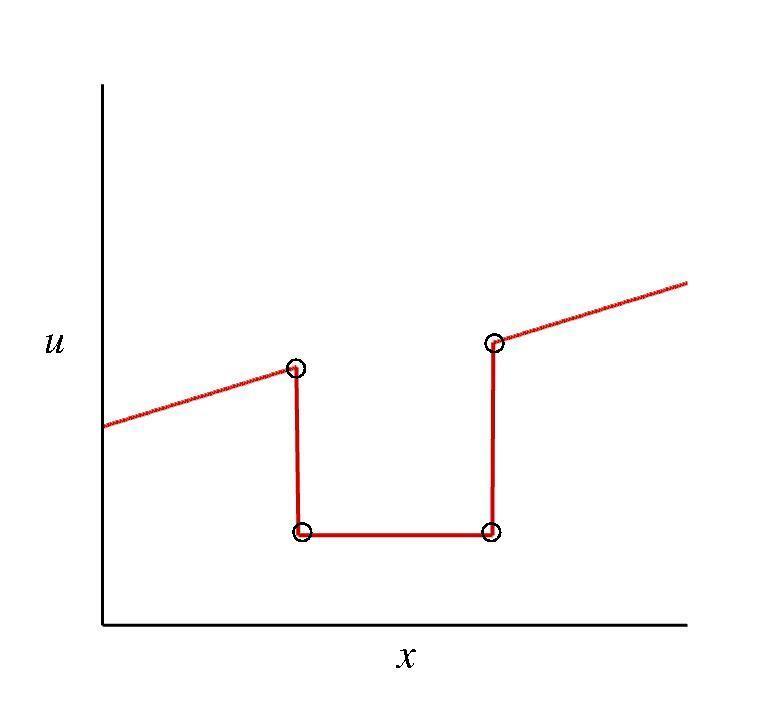}
 \caption{With a decreasing and an increasing shock, $u(x,t)$ cannot be a reparameterization. The circles represent the points chosen to prove that $u(x,t)$ cannot be a reparameterization of initial conditions.}
 \label{case2}
 \end{center}
 \end{figure}

\paragraph{Case 3}
Assume that the shock at $\chi_l(t)$ is an increasing shock and that
the shock at $\chi_r(t)$ is an increasing shock.  Since
$a(\chi_l(t_1)^-,t_1)=c(\chi_r(t_1)^-,t_1)$, for at least a short
period of time after $t_1$, the left value of $b(x,t)$ will be
greater than $a(\chi_l(t),t)$ and $c(\chi_l(t),t)$ and the right
value of $b(x,t)$ will be lower.  By choosing the points
$b(\chi_l(t)^+,t) > b(\chi_r(t)^-,t) < c(\chi_r(t)^+,t)$, $u(x,t)$
is not a reparameterization of the initial conditions.  See figure
\ref{case3}.

 \begin{figure}[!ht]
 \begin{center}
 \includegraphics[width=0.9\linewidth]{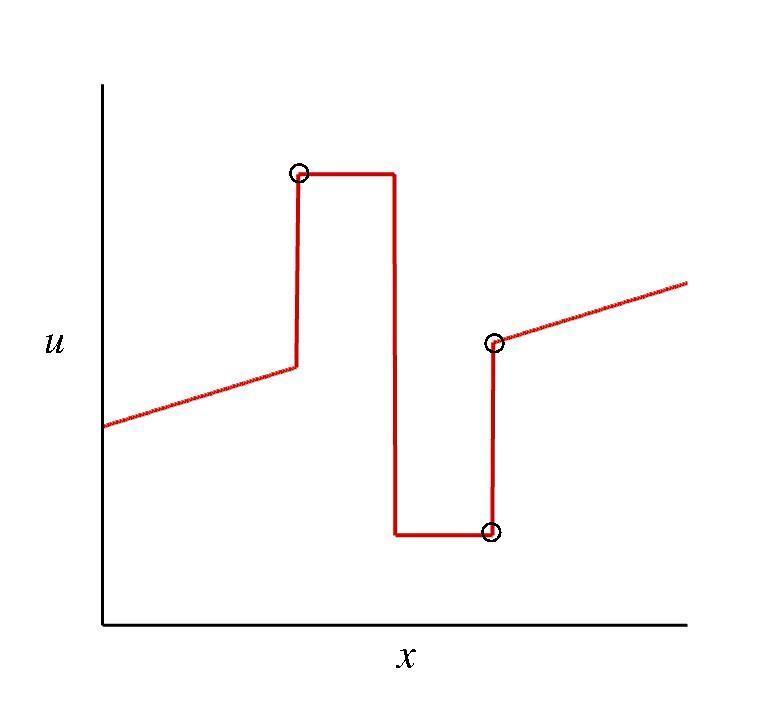}
 \caption{With an increasing shock and an increasing shock, $u(x,t)$ cannot be a reparameterization. The circles represent the points chosen to prove that $u(x,t)$ cannot be a reparameterization of initial
 conditions.}
 \label{case3}
 \end{center}
 \end{figure}

\paragraph{Case 4}
Assume that the shock at $\chi_l(t)$ is an increasing shock and that
the shock at $\chi_r(t)$ is a decreasing shock.  This case will be
divided into two subcases.  The first is that for all $x \in
(\chi_l(t),\chi_r(t))$ that $b(x,t)>a(\chi_l(t)^-,t)$ and
$b(x,t)>c(\chi_r(t)^+,t)$. If this is the case, then $u(x,t)$ is
proven to not be a weak solution by Lemma \ref{thatonelemma}.

 The second case is that there exists an $x_1 \in (\chi_l(t),\chi_r(t))$, such that $b(x_1,t) \leq a(\chi_l(t)^-,t)$ or $b(x_1,t) \leq c(\chi_r(t)^-,t)$.  Since $\chi_l(t)$ is an increasing shock and $\chi_r(t)$ is a decreasing shock, and $a(\chi_l(t_1)^-,t_1)=c(\chi_r(t_1)^-,t_1)$, for at least a short period of time after $t_1$, the left and right value of $b(x,t)$ will be greater than $a(\chi_l(t),t)$ and $c(\chi_l(t),t)$. Thus if $b(x_1,t) \leq a(\chi_l(t)^-,t)$ or $b(x_1,t) \leq c(\chi_r(t)^-,t)$, the points $b(\chi_l(t)^+,t)>b(x_1,t)<b(\chi_r(t)^+,t)$ show that $u(x,t)$ is not a reparameterization of initial conditions.  See figure
 \ref{case4}.

 \begin{figure}[!ht]
 \begin{center}
 \includegraphics[width=0.9\linewidth]{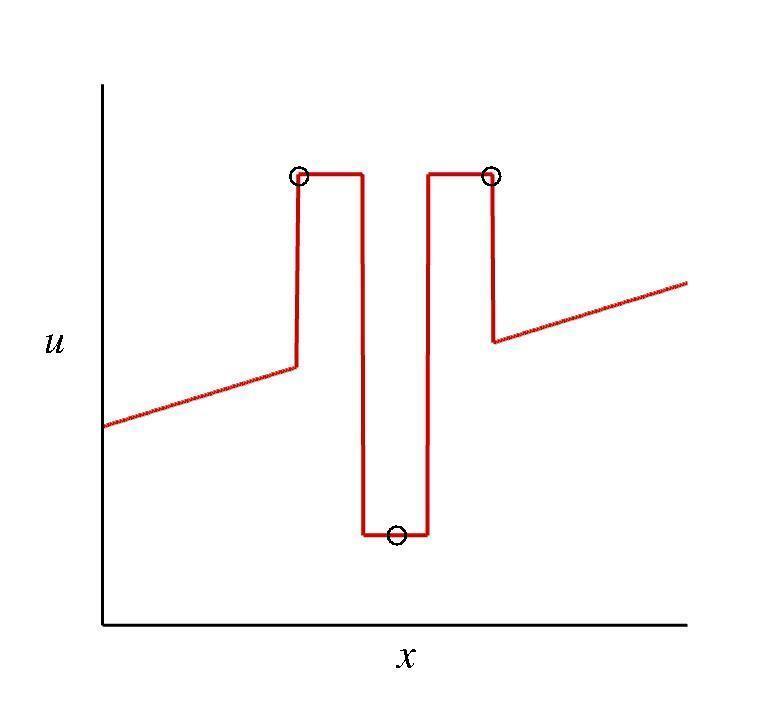}
 \caption{With a decreasing and an increasing shock, $u(x,t)$ cannot be a reparameterization. The circles represent the points chosen to prove that $u(x,t)$ cannot be a reparameterization of initial conditions.}
 \label{case4}
 \end{center}
 \end{figure}

Thus if $u(x,t)$ is a weak solution of the inviscid Burgers equation
and a reparameterization of initial conditions, it cannot engage in
spontaneous shock formation.

\subsubsection{Shock Splitting}\label{shocksplitting}

Assume that $u(x,t)$ is the entropy solution to the inviscid Burgers
equation up to time $t_1$ where an existing decreasing shock splits
into two or more at point $x^*$. Thus we say after time $t_1$,
\begin{equation}
u(x,t)=
\begin{cases}
a(x,t) & \text{if $          x <\chi_l(t)$,}\\
b(x,t) & \text{if $\chi_l(t) \le x < \chi_r(t)$}\\
c(x,t) & \text{if $\chi_r(t) \le x $},
\end{cases}
\end{equation}
where $a(x,t), b(x,t),$ and $c(x,t)$ are weak solutions to the
inviscid Burgers equation and $\chi_l(t)$ and $\chi_r(t)$ give the
position of the leftmost and rightmost discontinuities formed during
the shock splitting.  As there is assumed to be an already existing
decreasing shock at time $t_1$, $a(\chi_l(t_1)^-,t_1) >
c(\chi_r(t_1)^+,t_1)$.  There may be more than two shocks formed as
seen in Equation (\ref{spontanteousexample}), but we just need to
examine the leftmost and right most.

For $u(x,t)$ to be a weak solution we note that several things must
be true.  The speed of $\chi_l(t)$ and $\chi_r(t)$ are dictated by
the Rankine-Hugoniot jump conditions to be
\begin{equation}
\frac{d}{dt}\chi_l= \frac{a(\chi_l^-)+b(\chi_l^+)}{2}  \qquad
\frac{d}{dt}\chi_r= \frac{b(\chi_r^-)+c(\chi_r^+)}{2}.
\end{equation}
For there to be shock splitting, there must be some interval $(t_1,
t_2)$, where $\frac{d}{dt}\chi_r > \frac{d}{dt} \chi_l$.  Thus for
some interval $(t_1, t_2)$,  $a(\chi_l^-) > c(\chi_r^+)$ and thus
$b(\chi_l(t)^+,t)<b(\chi_r(t)^-,t)$.  Assume that $t \in (t_1, t_2)$
for the remainder of the subsection.

The shocks located at $\chi_l(t)$ and $\chi_r(t)$ must either be
increasing or decreasing shocks.  We will examine each of the
possibilities and show that each leads to $u(x,t)$ not being a
reparameterization of the initial conditions.

\paragraph{Case 1}
Assume that the shock at $\chi_l(t)$ is a decreasing shock and the
shock at $\chi_r(t)$ is a decreasing shock. Then
$a(\chi_l(t)^-,t)>b(\chi_l(t)^+,t)$ and $ b(\chi_r(t)^-,t)>
c(\chi_r(t)^+,t)$. We know that $b(\chi_l(t)^+,t)<b(\chi_r(t)^-,t)$,
and thus  $a(\chi_l(t)^-,t)>b(\chi_l(t)^+,t)<b(\chi_r(t)^-,t)$,
shows $u(x,t)$ is not a reparameterization of initial conditions.
See figure \ref{case1split}.

 \begin{figure}[!ht]
 \begin{center}
 \includegraphics[width=0.9\linewidth]{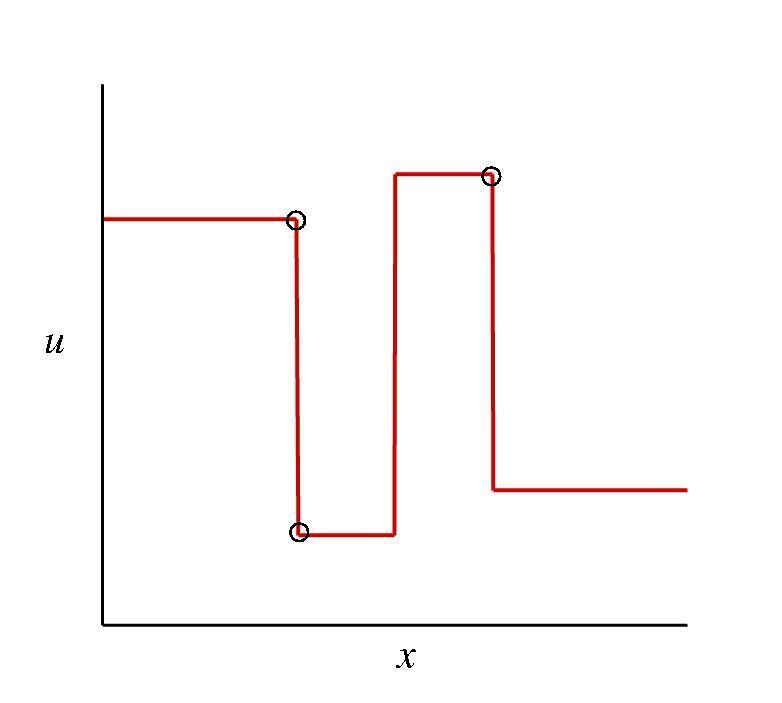}
 \caption{With two decreasing shocks, $u(x,t)$ cannot be a reparameterization. The circles represent the points chosen to prove that $u(x,t)$ cannot be a reparameterization of initial conditions.}
 \label{case1split}
 \end{center}
 \end{figure}

\paragraph{Case 2}
Assume that the shock at $\chi_l(t)$ is a decreasing shock and the
shock at $\chi_r(t)$ is an increasing shock.  Then
$a(\chi_l(t)^-,t)>b(\chi_l(t)^+,t)$ and $c(\chi_r(t)^+,t)>
b(\chi_r(t)^-,t)> b(\chi_l(t)^+,t)$.  Thus the points
$a(\chi_l(t)^-,t)>b(\chi_l(t)^+,t)<c(\chi_r(t)^+,t)$, shows $u(x,t)$
is not a reparameterization of initial conditions. See figure
\ref{case2split}.

 \begin{figure}[!ht]
 \begin{center}
 \includegraphics[width=0.9\linewidth]{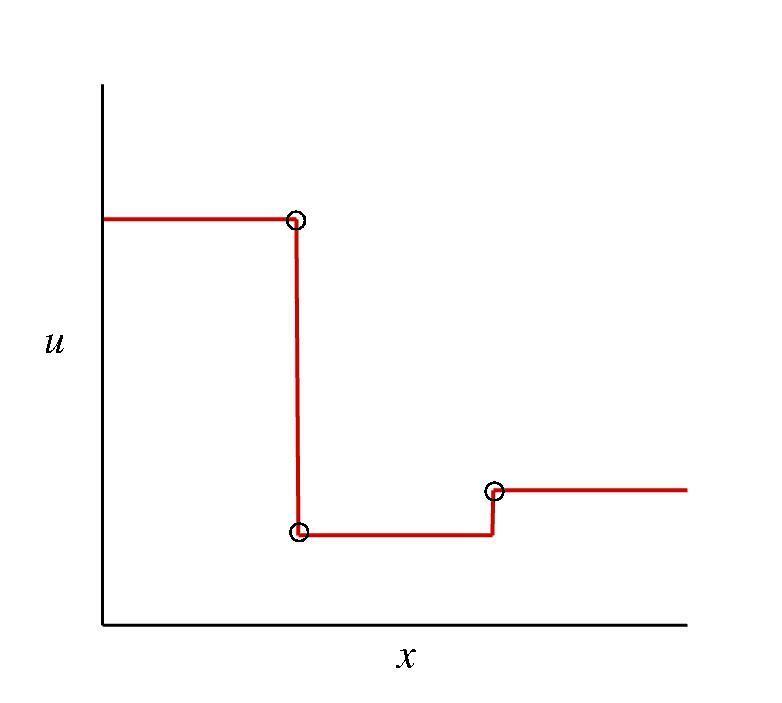}
 \caption{With a decreasing and an increasing shock, $u(x,t)$ cannot be a reparameterization. The circles represent the points chosen to prove that $u(x,t)$ cannot be a reparameterization of initial conditions. }
 \label{case2split}
 \end{center}
 \end{figure}

\paragraph{Case 3}
Assume that the shock at $\chi_l(t)$ is an increasing shock and that
the shock at $\chi_r(t)$ is an increasing shock.  Then
$a(\chi_l(t)^-,t) < b(\chi_l(t)^+,t)$ and $c(\chi_r(t)^+,t)>
b(\chi_r(t)^-,t)$.  From the Rankine-Hugoniot jump conditions, for
the interval $(t_1, t_2)$,  $a(\chi_l^-) > c(\chi_r^+)$ and thus
$b(\chi_l(t)^+,t)<b(\chi_r(t)^-,t)$.  This is a contradiction, so
$u(x,t)$ is not a weak solution.

\paragraph{Case 4}
Assume that the shock at $\chi_l(t)$ is an increasing shock and that
the shock at $\chi_r(t)$ is a decreasing shock.  This case will be
divided into two subcases.  The first is that for all $x \in
(\chi_l(t),\chi_r(t))$ that $b(x,t)>a(\chi_l(t)^-,t)$. If this is
the case, then $u(x,t)$ is proven to not be a weak solution by Lemma
\ref{thatonelemma}.

 The second case is that there exists an $x_1 \in (\chi_l(t),\chi_r(t))$, such that $b(x_1,t) \leq a(\chi_l(t)^-,t)$. Since $\chi_l(t)$ is an increasing shock $b(\chi_l(t)^+,t)>a(\chi_l(t),t)$ and thus $ b(\chi_l(t)^+,t)>b(x_1,t)$.  Since $b(\chi_l(t)^+,t)<b(\chi_r(t)^-,t)$, the points $ b(\chi_l(t)^+,t)>b(x_1,t)<b(\chi_r(t)^-,t)$ show that $u(x,t)$ is not a reparameterization of initial conditions.  See figure \ref{case4split}.

 \begin{figure}[!ht]
 \begin{center}
 \includegraphics[width=0.9\linewidth]{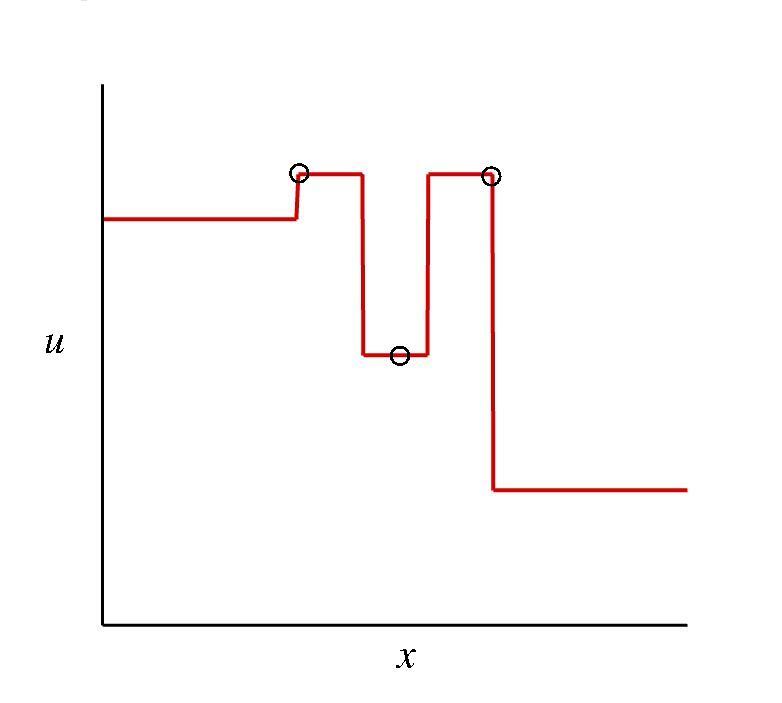}
 \caption{With a decreasing and an increasing shock,$u(x,t)$ cannot be a reparameterization. The circles represent the points chosen to prove that $u(x,t)$ cannot be a reparameterization of initial conditions.}
 \label{case4split}
 \end{center}
 \end{figure}

Thus if $u(x,t)$ is a weak solution of the inviscid Burgers equation
and a reparameterization of initial conditions, it cannot engage in
spontaneous shock formation or shock splitting.

\subsubsection{Entropy violating solutions are not reparameterizations of initial conditions}
From the previous sections the following lemma can be established.

\begin{lemma}\label{entropysolutionlemma}
Let $u(x,t)$ be a weak solution of the inviscid Burgers equation
where the initial conditions satisfy condition B. If $u(x,t)$ is
reparameterization of initial conditions then it is the entropy
solution.
\end{lemma}
\begin{proof}
 Clearly this follows from the results in Sections \ref{decreasingslopesection},
 \ref{spontaneousshocks}, and \ref{shocksplitting}.  \end{proof}

\subsection{Convergence to the Entropy Solution}
Based on Lemma \ref{entropysolutionlemma}, the following theorem
regarding the CFB equations converging to the entropy solution can
be established.

\begin{theorem}\label{entropytheorem}
The solutions $u^\alpha$ of the CFB equations converge to the
entropy solution of the inviscid Burgers equation for initial
conditions satisfying condition B.
\end{theorem}

\begin{proof}
 It has already been established that $u^\alpha$ converges to a
 weak solution of the inviscid Burgers equation, $u$, in $L^1_\text{loc}$.
 In the existence uniqueness proof, it was established that $u^\alpha$ is
 a reparameterization of initial conditions.  Clearly if every $u^\alpha$
 is a reparameterization, then its limit $u$ will also be a reparameterization
 of initial conditions.  Since $u$ is a  reparameterization of initial conditions and a weak solution to the
 inviscid Burgers equation, by Lemma \ref{entropysolutionlemma}, $u$
 must be the entropy solution.
 \end{proof}

Now that we have established that the solutions of the CFB equations
converge to the entropy solution of the inviscid Burgers equation
for initial conditions satisfying condition B. The following section
deals with how to regain the entropy solution for discontinuous
initial conditions and why we believe that this result hold true for
more general cases.

\section{Extension into discontinuous initial conditions}\label{moreentropysolutionsection}
Section \ref{entropysolutionsection} proves that the CFB equations
will converge to the entropy solution for a specific set of initial
conditions.  This section explains the intuitive reasoning on why it
is suspected that the CFB equations will converge to the entropy
solution for any continuous initial conditions and why it will not
for discontinuous initial conditions.  It then shows how the
equations can be changed slightly to incorporate discontinuous
initial conditions.  Begin by examining a commonly examined problem
for the inviscid Burgers equation.

\subsection{Example of entropic and non-entropic behavior for the inviscid Burgers equation}\label{behaviorexample}
Consider the initial conditions
\begin{equation}
u_0(x)=
\begin{cases}
0 & \text{if $x < 0 $}\\
1 &\text{if $x \ge 0$}.
\end{cases}
\end{equation}
The method of characteristics does not provide the value of $u$ in
the wedge $0<x<t$ as seen in figure \ref{emptywedge}.  The entropy
solution fills this wedge with the function $u(x,t)=\frac{x}{t}$
with characteristics fanning out from the original discontinuity as
seen in figure \ref{filledwedge}a.  This creates a rarefaction wave
and eliminates the discontinuity after time $t=0$.  A non-entropic
solution will allow the discontinuity to continue to exist.  It will
fill the wedge with new characteristics which continuously originate
from the discontinuity as time progresses as seen in figure
\ref{filledwedge}b.  Thus the non-entropic solution creates new
`information' as time progresses.

This problem embodies the essential behavior of entropic and
non-entropic solutions and provides the basis for our reasoning in
the following subsections.

 \begin{figure}[!ht]
 \begin{center}
 \includegraphics[width=0.9\linewidth]{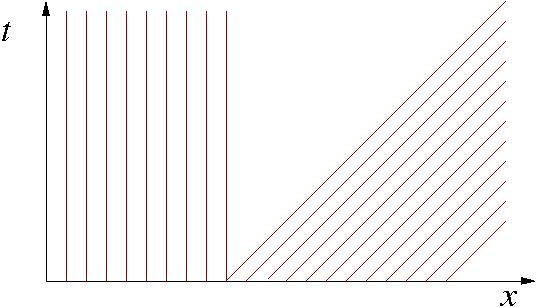}
 \caption{The area $0<x<t$ is not filled by characteristics}
 \label{emptywedge}
 \end{center}
 \end{figure}

  \begin{figure}[!ht]
\begin{center}
\begin{minipage}{0.48\linewidth} \begin{center}
  \includegraphics[width=.9\linewidth]{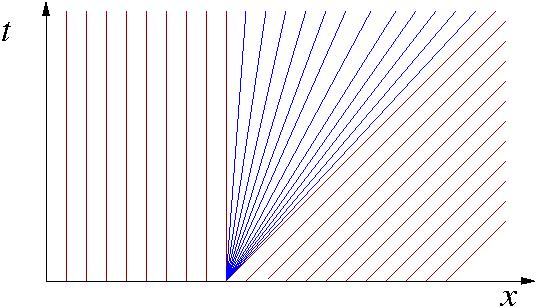}
\end{center} \end{minipage}
\begin{minipage}{0.48\linewidth} \begin{center}
  \includegraphics[width=.9\linewidth]{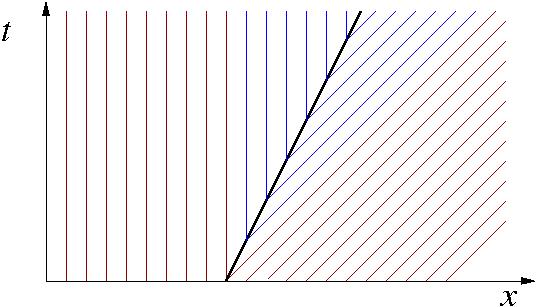}
\end{center} \end{minipage}\\
\begin{minipage}{0.48\linewidth}\begin{center} (a) \end{center} \end{minipage}
\begin{minipage}{0.48\linewidth}\begin{center} (b) \end{center}
\end{minipage}\vspace{-2mm}
\caption{This is how entropic and non-entropic solutions fill the
empty wedges.  a)  The entropic solution fills the wedge with a
rarefaction wave.  b) The non-entropic solution maintains the
discontinuity by creating new characteristics.}
 \label{filledwedge}
\end{center}
\end{figure}

\subsection{Convergence to entropy solution for all continuous initial conditions}
In section \ref{entropysolutionsection}, it was proven that for
initial conditions satisfying condition B, that the solutions to the
CFB equations converge to the entropy solution.  It is the
conjecture of this paper that the solutions to the CFB equations
converge to the entropy solution for all continuous initial
conditions.  As mentioned above a non-entropic solution will create
new characteristics, or new `information' as time progresses.  The
solutions to the CFB equation do not.  The existence and uniqueness
theorem proven in our previous paper, \cite{Norgard:08b} and
restated here as Theorem \ref{existencetheorem}, established that
the solution takes the form $u(x,t)=u_0(\phi(x,t))$, where
$\phi(x,t)$ is an increasing function of $x$ for any time, and
$\phi(x,0)=x$.  This shows that no new information is being created
in the CFB equations. Since the solutions to the CFB equations are
converging to a weak solution to the inviscid Burgers equation and
no new information is being created, it is reasonable to expect the
solutions to converge to the entropy solution.

\subsection{Nonconvergence for discontinuous initial conditions}
Consider initial conditions that have an increasing discontinuity in
them.  The entropy solution to the inviscid Burgers equation creates
a rarefaction wave from the discontinuity which takes on all the
values spanned by the discontinuity.  No new characteristics are
formed, as all originate from the discontinuity at time $t=0$, but
$u(x,t)$ now has values that did not originally exists in $u_0$.  As
shown in the existence and uniqueness theorem for the CFB equations,
the solutions to the CFB equations must have only the values found
in $u_0$.  Thus for initial conditions containing an increasing
discontinuity, the CFB equations will not converge to the entropy
solution.  An example of this can be found in our previous paper
\cite{Norgard:08b}, in section 6, where a traveling wave solution to
the CFB equations can be seen to converge to a non-entropic weak
solution. For this reason, we eliminate discontinuous initial
conditions from the admissible class of initial conditions.

\subsection{Conjecture}
Based on the reasoning in the previous two subsections, we present
the following conjecture.
\begin{conjecture}
\label{bigconjecture} The solutions $u^\alpha$ of the CFB equations
converge to the entropy solution of the inviscid Burgers equation
for continuous initial conditions  as $\alpha \to 0$.
\end{conjecture}
Assuming this conjecture is true, there is still the matter of
discontinuous initial conditions.  The following subsection creates
a new system that if the conjecture is true will converge to the
entropy solution for all bounded initial conditions.

\subsection{Regaining discontinuous initial conditions}
In regarding discontinuous initial conditions, begin by assuming
that for all $C^1$ initial conditions the solutions to the CFB
equations converge to the entropy solution.  Then if the  $C^1$
initial conditions limit to the discontinuous initial conditions in
$L^1_\text{loc}$, at the same time as $\alpha \to 0$ then the
solutions will converge to the entropy solution for the
discontinuous initial conditions.  To prove this we use a theorem
proven by Oleinik \cite{OleinikOA:57a} presented here.

\begin{theorem}\label{oleinik}
 Let $u^n(x,t)$ be the entropy solution for the inviscid Burgers equation with initial conditions $u^n(x,0)=u^n_0(x)$ and $u^n_0(x) \le m$ for all $n$.  Let
$$\int_{-\infty}^\infty f(x)\left[u^n_0(x)-u_0(x)\right] dx \to 0$$
for $n \to \infty$ for any compactly supported continuous function
$f(x)$.  Then the sequence $u^n(x,t)$ converges for $n \to \infty$
to the entropy solution $u(x,t)$ in $L^1_\text{loc}$ with initial
conditions $u(x,0)=u_0(x)$.
\end{theorem}

This theorem is employed in proving the following theorem.

\begin{theorem}\label{entropytheorem2}
 Let $u^{n,\alpha}$ be solutions to the CFB equations with initial conditions
 $u^{n,\alpha}(x,0)=u^n_0(x)$.  Let $u^n_0(x)$ converge to $u_0(x)$
 in $L^1$ as $n \to \infty$.  Let $u^e$ be the entropy solution to
 the inviscid Burgers equation with initial conditions  $u^e(x,0)=u_0(x)$.
 If for initial conditions $u_0(x) \in C^1$, the solutions to the CFB
 equations converge to the entropy solution in $L^1_\text{loc}$ ,
 then $u^{n,\alpha}$ converges to $u^e(x,t)$ in $L^1_\text{loc}$
 as  $n \to \infty$ and $\alpha \to 0$.
\end{theorem}

\begin{proof}
 Let $\Omega$ be a compact subset of $\mathbb{R} \times [0,T]$.  For $u^{n,\alpha}$ to converge to $u^e(x,t)$ in $L^1_\text{loc}$,
\begin{equation}
 \lim_{n\to \infty \,\, \alpha \to 0} \int \int_\Omega |u^{n,\alpha}-u^e| =0.
\end{equation}
Let $u^n(x,t)$ be the entropy solution to the inviscid Burgers
equation with initial conditions  $u^n(x,0)=u^n_0(x)$.  We have
assumed that
\begin{equation}
 \lim_{\alpha \to 0} \int \int_\Omega |u^{n,\alpha}-u^n| =0.
\end{equation}
From Theorem \ref{oleinik} we know that
\begin{equation}
 \lim_{n\to \infty} \int \int_\Omega |u^{n}-u^e| =0.
\end{equation}
Thus employing the triangle inequality we find
\begin{eqnarray}
 \lim_{n\to \infty \,\, \alpha \to 0} \int \int_\Omega |u^{n,\alpha}-u^e| &\le
\lim_{n\to \infty \,\, \alpha \to 0} \int \int_\Omega
|u^{n,\alpha}-u^n| \nonumber \\
& \qquad + \lim_{n\to \infty \,\, \alpha \to 0} \int \int_\Omega
|u^{n}-u^e|=0
\end{eqnarray}
 \end{proof}

Using Theorem \ref{entropytheorem2} it is easy to see that the
solutions to the initial value problem
\begin{align}
\label{newschemea}
u_t+(u \ast g^\alpha)u_x=0\\
\label{newschemeb} u(x,0)=u_0 \ast g^\alpha.
\end{align}
will converge to the entropy solution of the inviscid Burgers
equation with any initial condition $u_0(x)$ as $\alpha \to 0$. This
scheme can handle discontinuous initial conditions, providing a
greater usefulness.

\section{Numerics}\label{numericssection}
Section \ref{moreentropysolutionsection} proposes that Equations
(\ref{newschemea}) and (\ref{newschemeb}) are a new system for the
convectively filtered Burgers equation that is expected to converge
to the entropy solution of the inviscid Burgers equation as $\alpha
\to 0$ for all bounded initial conditions.  This section runs some
numerical simulation of the proposed system and shows evidence of
convergence to the entropy solution.

\subsection{The entropy solution}
The specific initial condition being examined is the indicator
function for the interval $(1,2)$ or
\begin{equation}
u_0(x)=
\begin{cases}
1 &\text{if $x \in (1,2) $}\\
0 & \text{otherwise}.
\end{cases}
\end{equation}
For the entropy solution to the inviscid Burgers equation, the right
side of the initial pulse will form the standard right traveling
shock and the left side will form a rarefaction wave.  At time
$t=2$, the rarefaction wave meets with the shock front the shock
front begins to decrease in amplitude and speed.  For time $t<2$ the
entropy solution for the given initial conditions is
\begin{equation}
u(x,t)=
\begin{cases}
0 &\text{if $x \leq 1$}\\
\frac{x-1}{t} &\text{if $x \in (1,1+t) $}\\
1 &\text{if $x \in (1+t,2+.5t) $}\\
0 & \text{if $x \ge 2+.5t$ }.
\end{cases}
\end{equation}
For time $t \ge 2$ the entropy solution  is
\begin{equation}
u(x,t)=
\begin{cases}
0 &\text{if $x \leq 1$}\\
\frac{x-1}{t} &\text{if $x \in (1,(2t)^\frac{1}{2}+1) $}\\
0 & \text{if $x \ge (2t)^\frac{1}{2}+1 $}.
\end{cases}
\end{equation}
It is to this solution that the CFB equations' solutions are
compared.

\subsection{Description of Numerical Methods}
Holm and Staley performed successful simulations of the CFB
equations with the Helmholtz filter, using a pseudospectral method
\cite{Holm:03a}. For this paper a very similar method is used.  With
the Helmholtz filter,   Equations (\ref{newschemea}) and
(\ref{newschemeb}) can be written as
\begin{equation}
\label{newschemeagaina} \frac{\p}{\p t}\ubar+\frac{\p}{\p x} \frac
{\ubar^2}{2}=-\frac{3}{2}\alpha^2 \left(I-\alpha^2 \frac{\p}{\p x}^2
\right)^{-1} \frac{\p}{\p x} (\ubar_x)^2
\end{equation}
\begin{equation}
\label{newschemeagainb} \ubar(x,0)=(u_0 \ast g^\alpha)\ast g^\alpha.
\end{equation}
It is these equations that are numerically simulated.

Equation (\ref{newschemeagaina}) is advanced through time with an
explicit, Runge-Kutta-Fehlberg predictor/corrector (RK45).  The
initial timestep is chosen low enough to achieve stability, and is
then varied by the code using the formula
\begin{equation}
h_{i+1}=\gamma h_{i} \left( \frac{\varepsilon h_{i}} {||\bar{u}_i
-\hat{u}_i||_{2} }\right)^\frac{1}{4}.
\end{equation}
Thus the new time step is chosen from the previous time step and the
amount of error between the predicted velocity, $\bar{u}$ and the
corrected velocity $\hat{u}$. The relative error tolerance was
chosen at $\varepsilon=10^{-4}$ and the safety factor $\gamma=0.9$.

Spatial derivatives and the inversion of the Helmholtz operator were
computed in the Fourier domain.  The velocity was converted into the
Fourier domain using a Fast Fourier Transform, multiplied by the
appropriate term and then converted back into the physical domain.
This pseudospectral method of calculating the derivative was chosen
to reduce artificial viscosity.

In Holm and Staley's method spatial derivatives were conducted using
a fourth-order finite difference and an artificial viscosity was
applied to the high wave modes to prevent aliasing errors
\cite{Holm:03a}.  Because the simulations are addressing convergence
to the entropy solution, as little artificial and numerical
viscosity as possible is desired.  For this reason derivatives were
done in the Fourier domain and no artificial viscosity was
introduced.

The simulations were done at the resolution of $2^{16}=65536$ grid
points.  Aliasing errors occurred, but did not introduce significant
amounts of error in the short time the simulations were run.  Figure
\ref{aliasing} shows the spectral energy of the simulation for
$\alpha=0.02$ at various times.  This was the worst case of aliasing
error and it can be seen that the error does not reach more than
approximately $10^{-13}$, at $t=3$ which is approximately 100,000
timesteps.   It should be noted that simulations using Holm and
Staley's artificial viscosity prevented this aliasing error, with
little noticeable effect on the solution.   Simulations with the
artificial viscosity have been conducted and produce the same
general results presented in the following sections.

\begin{figure}[!ht]
\begin{center}
\begin{minipage}{0.48\linewidth} \begin{center}
  \includegraphics[width=.9\linewidth]{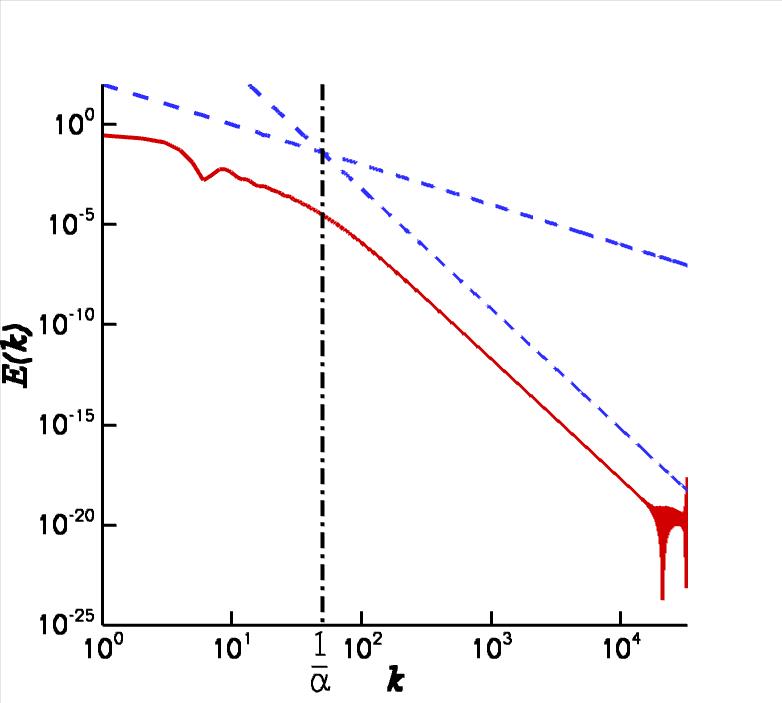}
\end{center} \end{minipage}
\begin{minipage}{0.48\linewidth} \begin{center}
  \includegraphics[width=.9\linewidth]{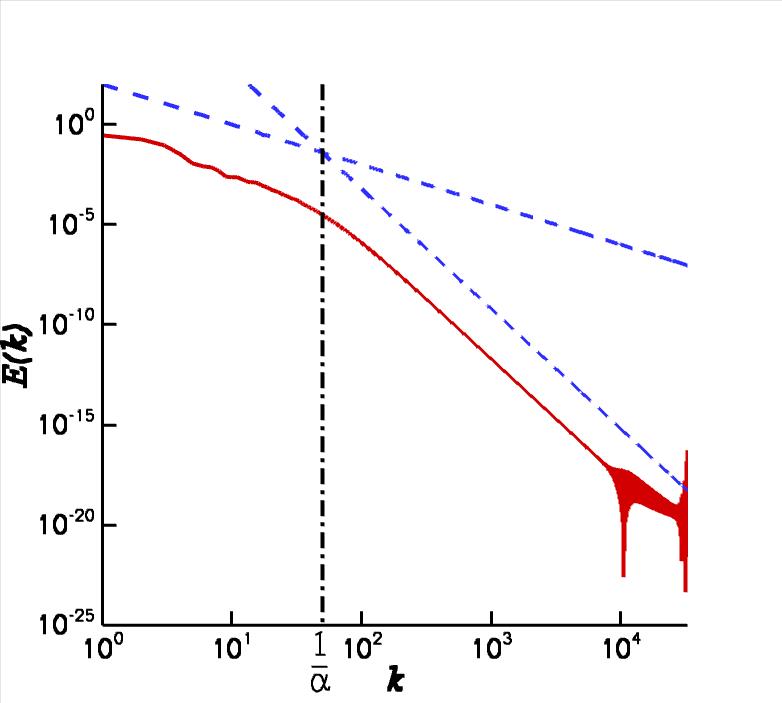}
\end{center} \end{minipage}\\
\begin{minipage}{0.48\linewidth}\begin{center} (a) \end{center} \end{minipage}
\begin{minipage}{0.48\linewidth}\begin{center} (b) \end{center}
\end{minipage}\vspace{2mm}
\begin{minipage}{0.48\linewidth} \begin{center}
  \includegraphics[width=.9\linewidth]{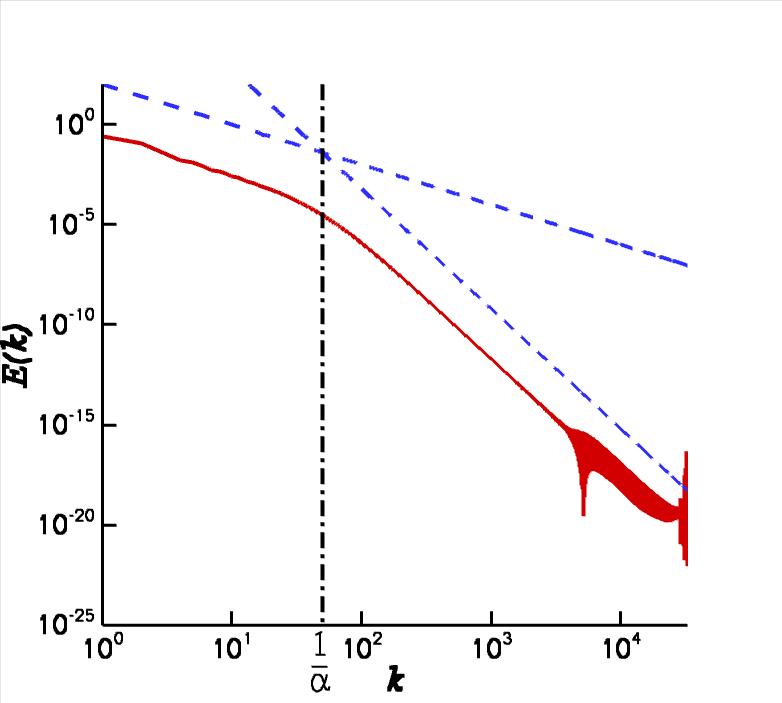}
\end{center} \end{minipage}
\begin{minipage}{0.48\linewidth} \begin{center}
  \includegraphics[width=.9\linewidth]{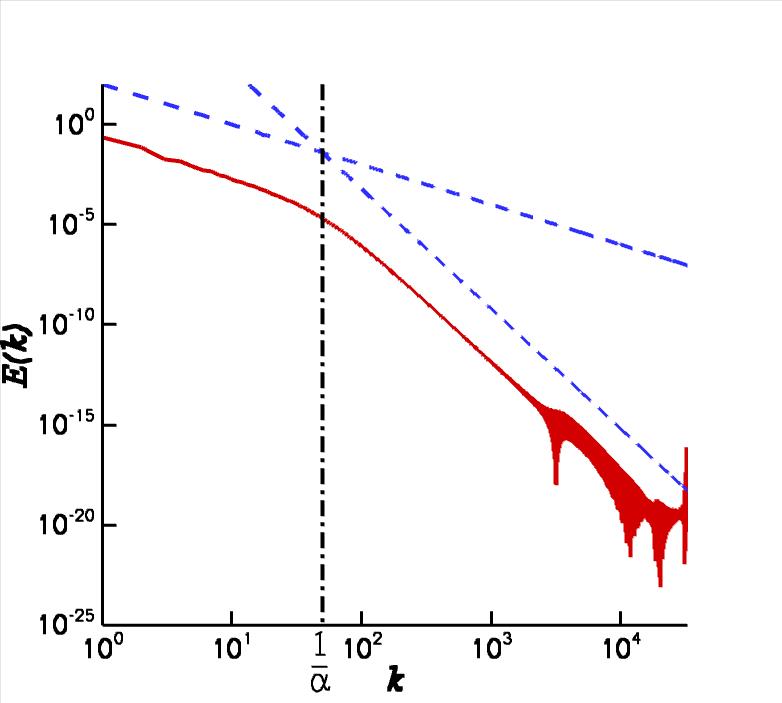}
\end{center} \end{minipage}\\
\begin{minipage}{0.48\linewidth}\begin{center} (c) \end{center} \end{minipage}
\begin{minipage}{0.48\linewidth}\begin{center} (d) \end{center}
\end{minipage}\vspace{-2mm}
\caption{The error caused by aliasing is visible in the lower part
of the energy spectrum. These are snapshots of the spectral energy
at times $t=0.5,1,2,3$. The aliasing error is seen to be propagating
up into lower wave modes, but at time $t=3$, the error is roughly
capped off by $10^{-13}.$  The dashed lines give a reference -$2$
and -$6$ slope.  The spectral energy slope changes at approximately
$\frac{1}{\alpha}$ as is expected.}
 \label{aliasing}
\end{center}
\end{figure}

 \subsection{Results}
 Nine different simulations were conducted with
 $\alpha=0.02, 0.03,...,0.10$.  The CFB equations showed
 behavior mirroring that of the entropy solution.  A traveling
 shock front and a rarefaction wave was seen.  Figure \ref{compare}
 compares the CFB simulations for $\alpha=0.02$ to the entropy
 solution at times $t=0,1,2,3$. In figure \ref{compare}a the
 difference in initial conditions can be seen with the entropy
 solution beginning with discontinuities and the CFB simulation
 having smoothed initial conditions.

 \begin{figure}[!ht]
\begin{center}
\begin{minipage}{0.48\linewidth} \begin{center}
  \includegraphics[width=.9\linewidth]{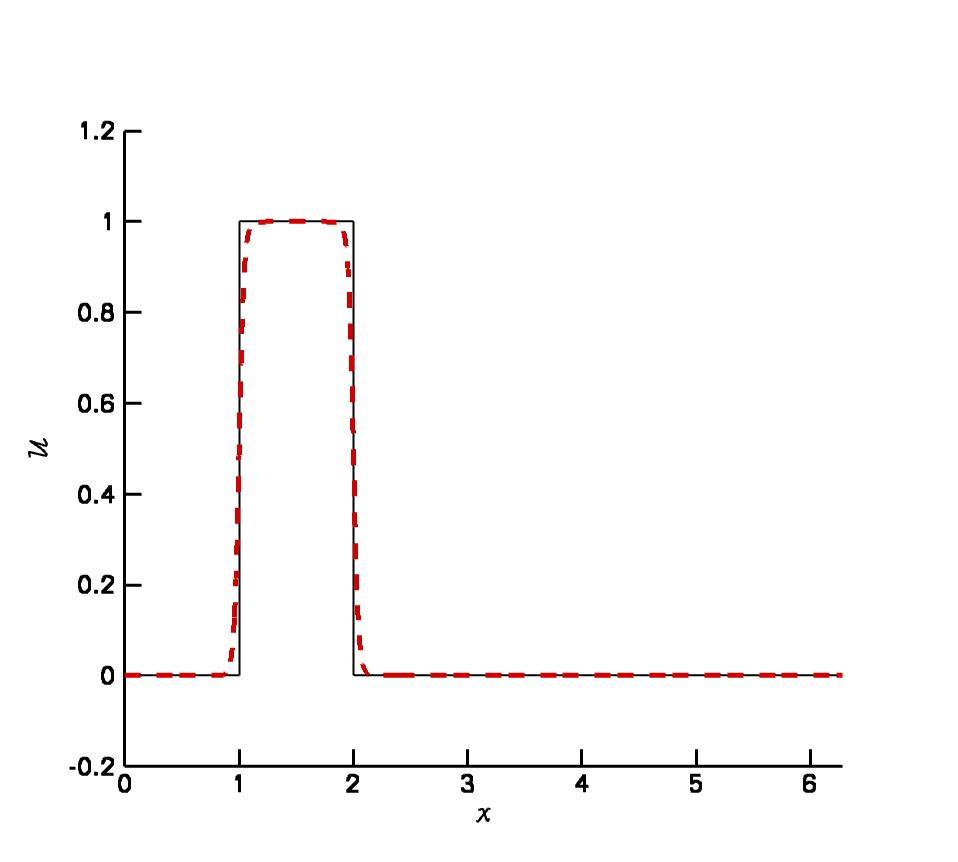}
\end{center} \end{minipage}
\begin{minipage}{0.48\linewidth} \begin{center}
  \includegraphics[width=.9\linewidth]{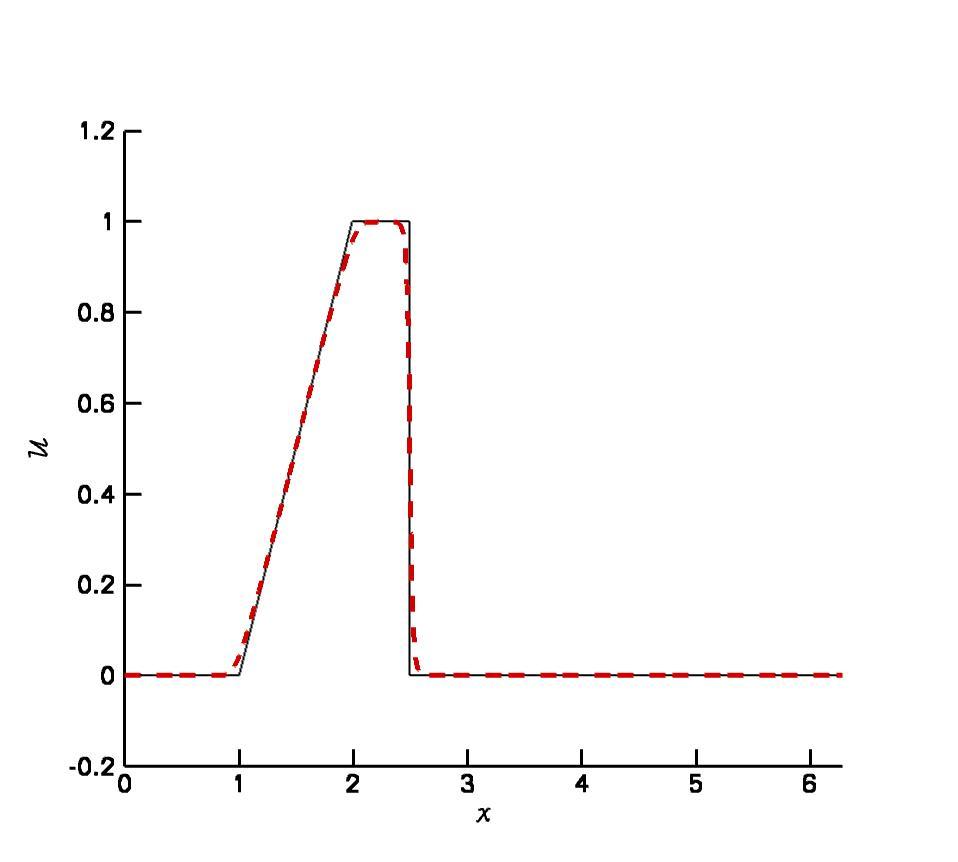}
\end{center} \end{minipage}\\
\begin{minipage}{0.48\linewidth}\begin{center} (a) \end{center} \end{minipage}
\begin{minipage}{0.48\linewidth}\begin{center} (b) \end{center}
\end{minipage}\vspace{2mm}
\begin{minipage}{0.48\linewidth} \begin{center}
  \includegraphics[width=.9\linewidth]{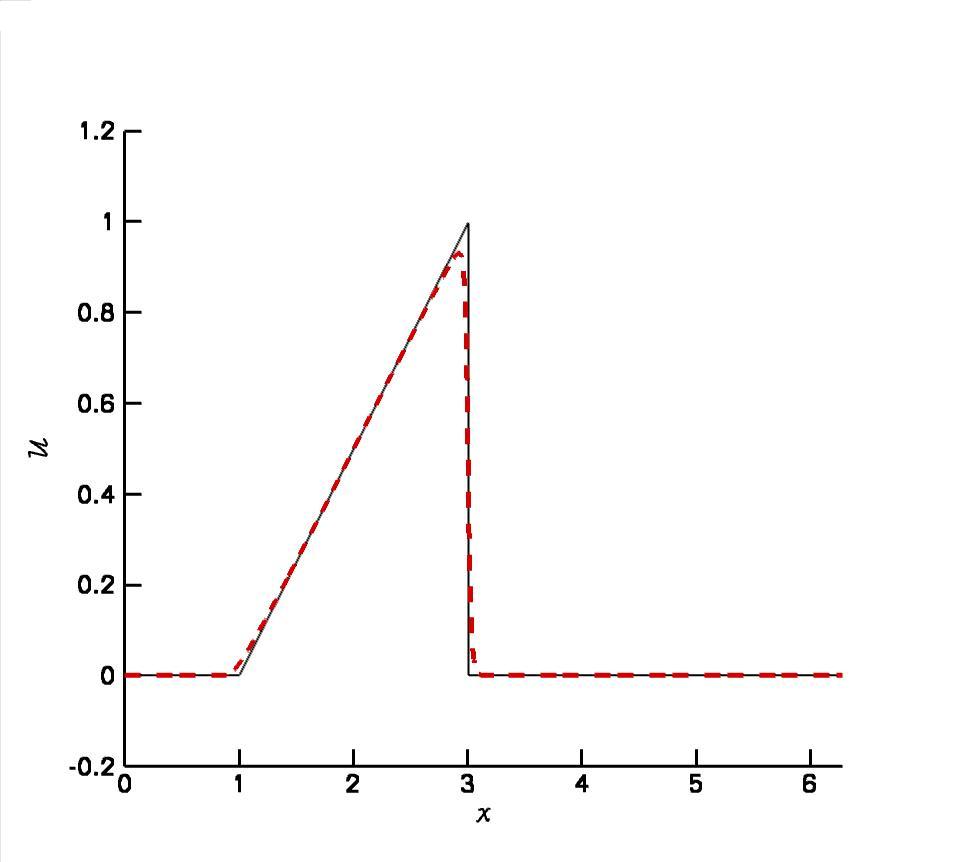}
\end{center} \end{minipage}
\begin{minipage}{0.48\linewidth} \begin{center}
  \includegraphics[width=.9\linewidth]{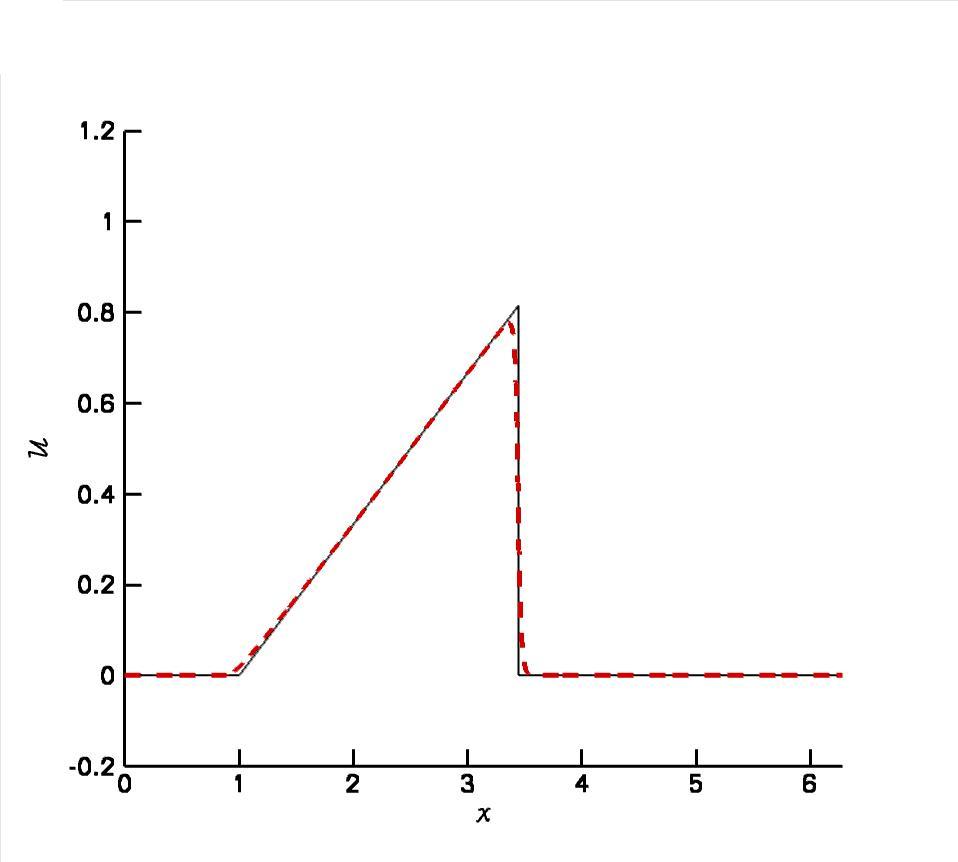}
\end{center} \end{minipage}\\
\begin{minipage}{0.48\linewidth}\begin{center} (c) \end{center} \end{minipage}
\begin{minipage}{0.48\linewidth}\begin{center} (d) \end{center}
\end{minipage}\vspace{-2mm}
\caption{This figure compares the entropy solution with the solution
to the CFB equations for $\alpha=0.02.$  It is easy to see that the
CFB equations' solution is capturing both the rarefaction wave and
the shock front behavior.}
 \label{compare}
\end{center}
\end{figure}

To evaluate the convergence of the CFB equations' solutions to the
entropy solution the $L_1$ norm of the error between the CFB
equations' solution and the entropy solution was taken.  Figure
\ref{l1error} plots $\alpha$ versus the error at times $t=0,1,2,3$.
At each time the error appears to be approaching zero linearly. Thus
numerical evidence suggests that the Equations (\ref{newschemea})
and (\ref{newschemeb}) will converge to the entropy solution of the
inviscid Burgers equation for initial conditions with
discontinuities.

 \begin{figure}[!ht]
 \begin{center}
 \includegraphics[width=0.9\linewidth]{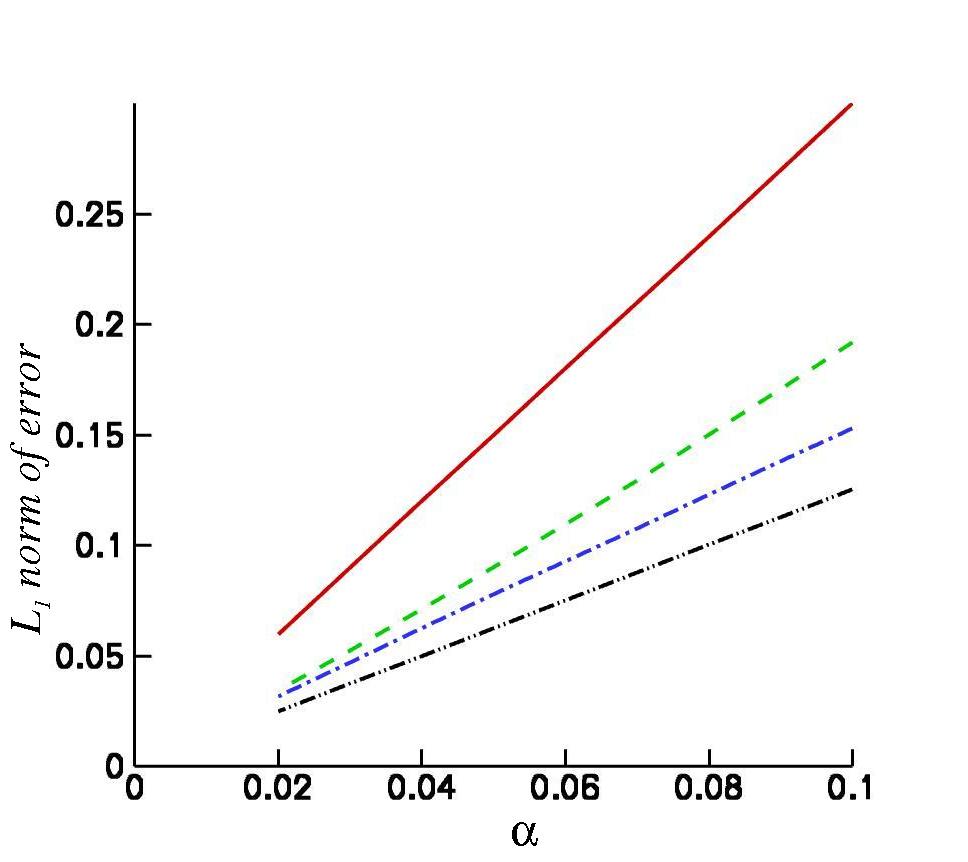}
 \caption{The $L_1$ norm of the error between the CFB equations'
 solution and the entropy solution.  The error is displayed for four
 different values of $t$.  $t=0$ \solid, $t=1$ \dashed, $t=2$ \dashdot,
 and $t=3$ \dashdotdot.  The error approaches zero roughly linearly
 as $\alpha \to 0$.}
 \label{l1error}
 \end{center}
 \end{figure}

\section{Conclusion}
Conservation laws can often have multiple weak solutions of which
there is one physically relevant solution, known as the entropy
solution.  It is important that any regularization of these
conservation laws  reflect the physical phenomenon they are meant to
address.  Thus it is important that the solutions to such
regularizations converge to the entropy solution.  The convectively
filtered Burgers equation has been shown to regularize the inviscid
Burgers equation.  This paper now shows that for a certain class of
initial conditions this regularization will converge to the entropy
solution.  It has also provided a method for extending this
convergence to a large class of initial conditions including
discontinuities.  These results are a crucial step in extending the
use of the convectively filtered method into popular use and perhaps
an extension into the Euler equations.

\section{Acknowledgments}
The research in this paper was partially supported by the AFOSR
contract FA9550-05-1-0334.

\bibliography{../RefA1}

\begin{thebibliography}{10}

\bibitem{Germano:91a}
M.~Germano, U.~Piomelli, P.~Moin, and W.H. Cabot.
\newblock A dynamic subgrid scale eddy viscosity model.
\newblock {\em Phys. Fluids A}, 3(7):1760--1765, 1991.

\bibitem{Lesieur:96a}
M.~Lesieur and O.~Metais.
\newblock New trends in large-eddy simulations of turbulence.
\newblock {\em \ARFM}, 28:45--82, 1996.

\bibitem{Hughes:01b}
T.J.R. Hughes, L.~Mazzei, and A.A. Oberai.
\newblock Large eddy simulation of turbulent channel flows by the variational
  multiscale method.
\newblock {\em Phys. Fluids}, 13(6):1784--1799, 2001.

\bibitem{FoiasC:02a}
C.~Foias, D.D. Holm, and E.S. Titi.
\newblock The three dimensional viscous {C}amassa-{H}olm equations, and their
  relation to the {N}avier-{S}tokes equations and turbulence theory.
\newblock {\em Journal of Dynamics and Differential Equations.}, 14(1):1--35,
  2002.

\bibitem{Marsden:01h}
J.E. Marsden and S.~Shkoller.
\newblock Global well-posedness of the {LANS}-$\alpha$ equations.
\newblock {\em Proc. Roy. Soc. London}, 359:1449--1468, 2001.

\bibitem{Mohseni:03a}
K.~Mohseni, B.~Kosovi\'{c}, S.~Shkoller, and J.E. Marsden.
\newblock Numerical simulations of the {L}agrangian averaged {N}avier-{S}tokes
  ({LANS}-$\alpha$) equations for homogeneous isotropic turbulence.
\newblock {\em \PF}, 15(2):524--544, 2003.

\bibitem{Chen:99b}
S.Y. Chen, D.D. Holm, L.G. Margoin, and R.~Zhang.
\newblock Direct numerical simulations of the {N}avier-{S}tokes-$\alpha$ model.
\newblock {\em Physica D}, 133:66--83, 1999.

\bibitem{Holm:04a}
A.~Cheskidov, D.D. Holm, E.~Olson, and E.S. Titi.
\newblock On a {Leray-$\alpha$} model of turbulence.
\newblock {\em Royal Society London, Proceedings, Series A, Mathematical,
  Physical \& Engineering Sciences}, 461(2055):629--649, 2004.

\bibitem{IlyinAA:06a}
A.A.Ilyin, E.M.Lunasin, and E.S.Titi.
\newblock A modified {L}eray-$\alpha$ subgrid scale model of turbulence.
\newblock {\em Nonlinearity}, 19:879–897, 2006.

\bibitem{Hanjalic:06a}
M.van Reeuwijk, H.J.J.Jonker, and K.~Hanjalić.
\newblock Incompressibility of the {L}eray-$\alpha$ model for wall-bounded
  flows.
\newblock {\em Physics of Fluids}, 18:1--4, 2006.

\bibitem{Mohseni:03c}
H.~S. Bhat, R.~C. Fetecau, J.~E. Marsden, K.~Mohseni, and M.~West.
\newblock Lagrangian averaging for compressible fluids.
\newblock {\em SIAM Journal on Multiscale Modeling and Simulation},
  3(4):818--837, 2005.

\bibitem{Norgard:08b}
G.~Norgard and K.~Mohseni.
\newblock A regularization of the {B}urgers equation using a filtered
  convective velocity.
\newblock {\em J. Phys. A: Math. Theor.}, 41:1--21, 2008.

\bibitem{WhithamGB:74a}
G.B. Whitham.
\newblock {\em Linear and nonlinear waves}.
\newblock John Wiley \& Sons, USA, 1974.

\bibitem{LaxPD:73a}
P.D. Lax.
\newblock {\em Hyperbolic Systems of Conservation Laws and the Mathematical
  Theory of Shock Waves}.
\newblock SIAM, Philadelphia, 1973.

\bibitem{OleinikOA:57a}
O.A. Oleinik.
\newblock Discontinuous solutions and non-linear differential equations.
\newblock {\em Am. Math. Soc. Transl.}, 26:95--172, 1957.

\bibitem{KruzkovSN:70a}
S.N Kruzkov.
\newblock First order quasilinear equations in several independent variables.
\newblock {\em Math. USSR Sbornik}, 10:217--243, 1970.

\bibitem{LellisCD}
Camillo~De Lellis.
\newblock Minimal entropy conditions for {Burgers} equation.
\newblock {\em Quarterly of Applied Mathematics}, 64(4):687--700, 2004.

\bibitem{BhatHS:06a}
H.S. Bhat and R.C. Fetecau.
\newblock A {Hamiltonian} regularization of the {Burgers} equation.
\newblock {\em J. Nonlinear Sci.}, 16(6):615--638, 2006.

\bibitem{Marsden:94a}
J.E. Marsden and T.S. Ratiu.
\newblock {\em Introduction to Mechanics and Symmetry}.
\newblock Springer-Verlag, New York, Inc., 1994.

\bibitem{Mohseni:06l}
K.~Mohseni, H.~Zhao, and J.~Marsden.
\newblock Shock regularization for the {B}urgers equation.
\newblock AIAA paper 2006-1516, 44$^\textrm{th}$ AIAA Aerospace Sciences
  Meeting and Exhibit, Reno, Nevada, January 9-12 2006.

\bibitem{Mohseni:07e}
G.~Norgard and K.~Mohseni.
\newblock A regularization of {Burgers} equation using a filtered convective
  velocity.
\newblock AIAA paper 2007-0714, 45$^\textrm{th}$ AIAA Aerospace Sciences
  Meeting and Exhibit, Reno, Nevada, January 8-11 2007.

\bibitem{Mohseni:07s}
G.~Norgard and K.~Mohseni.
\newblock Convectively filtered {Burgers} in one and multiple dimensions.
\newblock AIAA paper 2007-4221, 37$^\textrm{th}$ AIAA Fluid Dynamics Conference
  and Exhibit, Miami, FL, June 25 - 28 2007.

\bibitem{Lax:83a}
P.~D. Lax and C.D. Levermore.
\newblock The small dispersion limit of the {KdV} equations. iii.
\newblock {\em Comm. Pure Appl. Math.}, XXXVI:809--830, 1983.

\bibitem{Kawahara:70b}
T.~Kakutani and T.~Kawahara.
\newblock Weak ion-acoustic shock waves.
\newblock {\em J. Phys. Soc. Japan}, 29(4):1068--1073, 1970.

\bibitem{GurevichAV:74a}
A.V. Gurevich and L.P. Pitaevskii.
\newblock Nonstationary structure of a collisionless shock wave.
\newblock {\em Sov Phys JETP}, 38(2):291--297, 1974.

\bibitem{BressanA:00a}
A.Bressan.
\newblock {\em Hyperbolic Systems of Conservation Laws}.
\newblock Oxford University Press, 2000.

\bibitem{SerreD:99a}
D.Serre.
\newblock {\em Systems of Conservation Laws 1}.
\newblock Cambridge University Press, 1999.

\bibitem{DuoandikoetxeaJ:99a}
J.Duoandikoetxe.
\newblock {\em Fourier Analysis}.
\newblock AMS, 2000.

\bibitem{DonsigAP:02a}
K.R. Davidson and A.P. Donsig.
\newblock {\em Real Analysis with Real Applications}.
\newblock Prentice Hall, 2002.

\bibitem{DafermosCM:99a}
Constantine~M. Dafermos.
\newblock {\em Hyperbolic Conservation Laws in Continuum Physics}.
\newblock Springer, 2005.

\bibitem{Holm:03a}
D.D. Holm and M.F. Staley.
\newblock Wave structures and nonlinear balances in a family of evolutionary
  {PDE}s.
\newblock {\em SIAM J. Appl. Dyn. Syst.}, 2:323--380, 2003.

\end{thebibliography}
\bibliographystyle{unsrt}

\end{document}